\documentclass{jocg}
\usepackage{amsmath,amsfonts}
\usepackage{fancyhdr}

\usepackage{hyperref}
\usepackage{accents}
\usepackage{cite}
\usepackage{latexsym}
\usepackage{amssymb}
\usepackage{graphicx}
\usepackage{colortbl}

\usepackage{amsthm}
\theoremstyle{plain}
\newtheorem{theorem}{Theorem}
\newtheorem{lemma}{Lemma}
\newtheorem{corollary}{Corollary}
\newtheorem{proposition}{Proposition}
\newtheorem{question}{Question} 
\theoremstyle{definition}

\newenvironment{cpf}{\begin{trivlist} \item[] {\em Proof of Claim.}}{\hspace*{\stretch{1}} \end{trivlist}}

\usepackage[normalem]{ulem} 
\newcommand{\quotes}[1]{``#1''}

\def\dist{\mathop{\rm dist}}
\def\conv{\mathop{\rm conv}}

\DeclareMathOperator{\ar}{\rho}
\DeclareMathOperator{\disk}{Disk}

\def\Bscr{{\cal B}}  \def\Pscr{{\cal P}} 
\def\Cscr{{\cal C}}   
   
  \def\Sscr{{\cal S}} 
\def\Fscr{{\cal F}}

\newcommand{\Bcirc}{\overline{B}_2}

\usepackage{marginnote}
\usepackage[textwidth=2cm]{todonotes}

\title{\MakeUppercase{Packing, Hitting, and Colouring Squares}}

\author{
Marco Caoduro %
  \thanks{\affil{Sauder School of Business, The University of British Columbia, Vancouver, Canada}, 
          \email{marco.caoduro@ubc.ca}}\,
and Andr\'as Seb\H{o} %
  \thanks{\affil{CNRS, Laboratoire G-SCOP, Univ.~Grenoble Alpes, Grenoble, France},
          \email{andras.sebo@cnrs.fr}}\
}

\begin{document}
\maketitle

\pagestyle{plain}

\begin{abstract}
Given a finite family of squares in the plane,  the \emph{packing problem} asks for the maximum number $\nu$ of pairwise disjoint squares among them, while the \emph{hitting problem} for the minimum number $\tau$  of points hitting all of them. Clearly, $\tau \ge \nu$. Both problems are known to be NP-hard, even for families of axis-parallel unit squares.

The main results of this work provide the first non-trivial bounds for the $\tau / \nu$ ratio for not necessarily axis-parallel squares.
We establish an upper bound of $6$ for unit squares and $10$ for squares of varying sizes. The worst ratios we can provide with examples are $3$ and $4$, respectively.
For comparison, in the axis-parallel case, the supremum of the considered ratio is in the interval $[\frac{3}{2},2]$ for unit squares and $[\frac{3}{2},4]$ for squares of varying sizes.
The methods we introduced for the $\tau/\nu$ ratio can also be used to relate the chromatic number $\chi$ and clique number $\omega$ of squares by bounding the $\chi/\omega$  ratio by $6$ for unit squares and $9$ for squares of varying sizes.

The $\tau / \nu$ and $\chi/\omega$ ratios have already been bounded before by a constant for ``fat'' objects,  the fattest and simplest of which are disks and squares.
However, while disks have received significant attention, specific bounds for squares have remained essentially unexplored. This work intends to fill this gap.
\end{abstract}


\section{Introduction}\label{sec:intro}
Let $\Fscr$ be a finite family of convex sets in the plane. A {\em packing} in $\Fscr$ is a subfamily of pairwise disjoint sets in $\Fscr$, a {\em hitting set} of $\Fscr$ is a set of points which has a non-empty intersection with each $F\in\Fscr$,  a {\em colouring} of $\Fscr$ is a partition of $\Fscr$ into packings, and a {\em clique} in $\Fscr$ is a pairwise intersecting subfamily in  $\Fscr$.
The maximum size of a packing in $\Fscr$, the {\em packing number},  is denoted by $\nu(\Fscr)$ and the minimum size of a hitting set of $\Fscr$, the {\em hitting number},  by $\tau (\Fscr)$. 
The maximum size of a clique in $\Fscr$, the {\em clique number}, is denoted by $\omega(\Fscr)$ and the minimum size of a colouring of $\Fscr$, the {\em chromatic number}, by $\chi(\Fscr)$. 
For $p\in \mathbb{R}^2$ the number of $F\in\Fscr$ containing $p$ is the \emph{degree} of $p$ in $\Fscr$ and $\Delta(\Fscr)$ denotes the maximum degree.  
Clearly, $\nu(\Fscr) \le \tau(\Fscr)$ and  $\Delta(\Fscr) \leq \omega(\Fscr) \le \chi(\Fscr)$.
A family of convex sets is said to have the {\em Helly property} if, for any clique, there exists a point contained in all of them. If a family $\Fscr$ has the Helly property, then $\omega(\Fscr)=\Delta(\Fscr)$.  Axis-parallel rectangles (and boxes in arbitrary dimensions) have the Helly property. On the other hand, the property does not extend to general rectangles, not even to unit squares.  Figure \ref{fig:nu1tau2} shows three unit squares with $\Delta=2$ and $\omega=3$.

\begin{figure}
    \centering
    \includegraphics[scale=0.5]{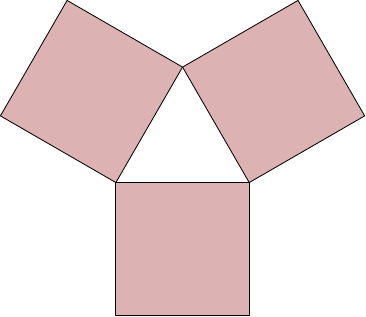}
    \caption{Three pairwise intersecting unit squares.}
    \label{fig:nu1tau2}
\end{figure}
The problems of determining the packing and hitting numbers of a given family of subsets in the plane are proved to be NP-hard already in the particular case of axis-parallel unit squares \cite{1981_Fowler}.
However, there are \emph{polynomial-time approximation schemes} (PTAS) 
to compute these parameters for arbitrary families of squares, even without the constraint of being axis-parallel \cite{2003_Chan}.

Imai and Asano \cite{1981_Imai} noticed that the clique number of a family of axis-parallel rectangles (equal to the maximum degree of the family) can be computed in linear time, but determining the chromatic number for such a family is NP-hard. Even though their hardness reduction does not seem to apply to axis-parallel squares directly, a simple modification of a result on unit disks of Clark, Colbourn, and Johnson's \cite[Theorem 2.1]{1990_Clark},  shows that this problem is NP-hard even for the special case of axis-parallel unit squares 
(see Appendix~\ref{sub:unitsquares}).
In this paper, our foci are the $\tau/\nu$ and $\chi/\omega$ ratios and not these complexity issues, which, however, are important since they make the existence of min-max theorems and exact algorithms unlikely unless P=NP.

It is surprising that the most natural bounds relating the packing and hitting numbers are wide open even for some of the simplest geometric objects. Our goal is to decrease this gap for squares.
We bound $\tau$ from above with a linear function of $\nu$ and deduce similar bounds for $\omega$ and $\chi$ as well.
These bounds are realized by using novel arguments that build on elementary geometric tools. The phenomenon is reminiscent to the \quotes{rounding} idea of integer programming. Indeed, it consists in ``filling holes for free'' in reformulations as ``covering'' problems, but instead of integrality,  the ``discrete jump" is due to the forcing rules of Euclidean geometry.
Using less refined and more common techniques, we also extend the linear bounds to a family of {\em similar convex sets}:  two convex sets are called {\em similar} if they arise from one another by translations, rotations, and homotheties. 
Similarity is an equivalence relation, and squares form an equivalence class. Better results can be proved if only translations and rotations are allowed, including the special case of unit squares. Excluding rotations, we get another special case that includes axis-parallel squares, for which even better results hold (see Table~\ref{tab:tau_bounds}).

Note that while linear bounds for various convex sets are known -- either through Pach's work~\cite{1980_Pach} on the $\chi/\nu$ ratio or via an extension of his method (Theorem~\ref{thm:hit_convexsets}) for the $\tau/\nu$ ratio -- our main achievement lies in providing the first specific tools and non-trivial bounds for squares.
This fills a major gap in the subject: although considerable efforts have been made to improve the upper and lower bounds for these ratios for disks~\cite{1990_Clark, 2004_Kim,1955_Hadwiger}, there seem to be no results specific to squares besides the statement of the relatively easy bounds for axis-parallel squares~\cite{2006_Ahlswede}.

\begin{table}
\centering
\begin{tabular}{c|c|c|c|c}
 & \textit{convex set} & \textit{centrally symmetric} & \textit{disk} & \textit{square} \\ \hline
\textit{translation} & $\tau \leq 6\nu$   &  $\tau \leq 6\nu$  & $\tau \leq 4\nu - 1$ & $\tau \leq 2\nu - 1$  \\ 
& \cite{2011_Dumitrescu}  &   \cite{2011_Dumitrescu}  &  \cite{2011_Dumitrescu} &  \cite{2006_Ahlswede} \\
\hline
\textit{translation +  } & $\tau \leq 16\nu$  &  $\tau \leq 7\nu$  & $\tau \leq 7\nu - 3$  & $\tau \leq 4\nu - 3$  \\
\textit{homothety} & \cite{2006_Kim} &  \cite{2011_Dumitrescu} & \cite{2011_Dumitrescu} & \cite{2006_Ahlswede} \\ \hline
\textit{translation +} & $\tau \leq 18 \ar ^2 \nu$ & $\tau \leq 4 \lceil  \ar    \rceil^2 \nu$ & $\tau \leq 4\nu - 1$ &  $\tau \leq 6 \nu$ \\ 
\textit{rotation} & Cor. \ref{cor:trans_rot_homothetic} & Cor. \ref{cor:trans_rot_homothetic}  & \cite{2011_Dumitrescu} & Thm. \ref{thm:hitting}  \\ \hline
\textit{translation +} & $\tau \leq 18 \ar ^2 \nu$ &  $\tau \leq 8\lceil  \ar    \rceil^2 \nu$ & $\tau \leq 7\nu - 3$ & $\tau \leq 10 \nu$ \\ 
\textit{homothety +} & Cor. \ref{cor:trans_rot_homothetic} &  Cor. \ref{cor:trans_rot_homothetic} & \cite{2011_Dumitrescu} & Thm. \ref{thm:hitting}  \\ 
\textit{rotation} &  &  &  &  \\
\end{tabular}
\caption{$\tau/\nu$ bounds for a family obtained by translations, rotations, or homotheties of a convex set $A$. The {\em slimness} $\ar(A)$ of  $A$  is defined here as  $R/r$ where $R$ is the smallest radius of a disk containing $A$ and $r$ is the largest radius of a disk contained in $A$ (see Section~\ref{sec:cover}).}
\label{tab:tau_bounds}
\end{table}

The numbers $\tau$, $\nu$, $\chi$, $\omega$ can be often interpreted and used as parameters of the intersection graph:

The {\em intersection graph} of a family $\Fscr$ of sets, denoted by $G(\Fscr)$, is the graph having $\Fscr$ as vertex set and an edge between two vertices if and only if the corresponding sets intersect.
Our notations  and terminology for $\Fscr$  follow the usual graph theory notations for  $G(\Fscr)$, that is, $\chi(\Fscr)=\chi(G(\Fscr))$ and $\omega(\Fscr)= \omega(G(\Fscr))$. Of course, the maximum degree $\Delta(G(\Fscr))$ of the \emph{graph} $G(\Fscr)$ is not the same as the maximum degree $\Delta(\Fscr)$ of the \emph{family} $\Fscr$.

A {\em neighbour} of  $F\in\Fscr$  is a set $F'\in\ \setminus \{F\}$ such that $F\cap F'\ne\emptyset$. 
The \emph{neighbourhood} $N(F)$ of $F$ consists of all its neighbours, while the {\em closed neighbourhood} is $N[F]:=N(F) \cup \{F\}$.
Given a positive integer $D\in \mathbb{N}$, a graph is called {\em $D$-degenerate} if each of its subgraphs has a vertex of degree at most $D$.  
We call a family $\Fscr$ of sets {\em $D$-degenerate} if every non-empty subfamily $\Fscr'\subseteq \Fscr$  contains a set $F'\in \Fscr'$ which intersects at most $D$ other sets of $\Fscr'$, that is, if its intersection graph is $D$-degenerate.
All vertices of a $D$-degenerate graph can be deleted by sequentially deleting a vertex of degree at most $D$; colouring the vertices in the reverse order, we see that the graph is $D+1$ colourable. 
Consequently,{\em a $D$-degenerate family has a $(D+1)$-colouring}.
Although this is quite a rough, greedy way of colouring, degeneracy is a frequently used method, often providing the best-known way to colour geometric intersection graphs \cite{1980_Pach, 2003_Perepelitsa, 2004_Kim}.

Clearly, \emph{deleting $N[F]$ $(F\in\Fscr)$,  the maximum size of a packing in $\Fscr$  decreases by at least one, and the hitting number decreases by at most  $\tau(N[F])$.}
This provides an inductive argument-based bounding of the $\tau/\nu$ ratio, leading to a linear upper bound whenever $\tau(N[F])$ can be upper bounded by a constant (see Lemma~\ref{lem::induction_idea}).
In close analogy with  $D$-degeneracy for colouring, we say that a family $\Fscr$ is {\em hitting-$k$-degenerate} ($k \in \mathbb{N}$) if for every $\Fscr' \subseteq \Fscr$  there exists $F'\in \Fscr'$ such that $\tau(N[F']) \leq k$. 
Hitting-degeneracy was successfully employed in Kim, Nakprasit, Pelsmajer, and Skokan~\cite{2006_Kim} and Dumitrescu and Jiang \cite{2011_Dumitrescu} when dealing with translates of convex bodies (referred to as \quotes{greedy decomposition}) and will also be our main framework for bounding the $\tau/\nu$ ratio. This framework is far from being optimal, but it is the unique one for most geometric intersection problems.
The following lemma gives a bound on $\tau(N[F])$ for squares. Its proof, presented in Section~\ref{sec:proofs}, is based on novel geometric ideas:

\begin{lemma} \label{lem:hitting_neighbours}
    Let $\Cscr$ be a family of unit squares.
    The neighbours of any square $C \in \Cscr$ can be hit by $10$ points.
    Moreover, if the centre of $C$ is left-most among all centres in $\Cscr$, $6$ points suffice.
\end{lemma}

While the induction by degeneracy is a kind of simple greedy framework,  bounding $\tau(N[F])$ is a real challenge. 
The first part of Lemma~\ref{lem:hitting_neighbours} can be extended to squares of arbitrary size by selecting a square with minimal size and applying homothety to each of its neighbours (\quotes{local homothety,} see Section \ref{sec:cover}),  allowing  us to conclude: 

 \begin{theorem}\label{thm:hitting}
 If $\Cscr$ is a family of squares, $\tau(\Cscr)\le 10\nu(\Cscr)$.
 Moreover, if the squares have equal size, $\tau(\Cscr)\le 6\nu(\Cscr)$. 
\end{theorem}

Finding lower bounds for the ratio $\tau/\nu$ is also challenging. The only known lower bound for families of axis-parallel squares is $3/2$, achieved by a family of unit squares whose intersection graph is a vertex disjoint union of $5$-cycles; no better lower bound is known for squares of different sizes. 
If arbitrary rotations of the squares are allowed, the  $\tau/\nu$ ratio for unit squares may even be $3$, and $4$ if squares of different sizes are allowed: 

\begin{theorem}\label{thm:lower} There exists a family of $9$ pairwise intersecting unit squares that cannot be hit with less than $3$ points.
Moreover, there exists a family of $13$ pairwise intersecting squares that cannot be hit with less than $4$ points.
\end{theorem} 

Pach \cite{1980_Pach} proved that for any family $\Fscr$  of convex sets in the plane,  $\chi(\Fscr) \leq 9q \Delta(\Fscr),$
where for each $F\in\Fscr$  the ratio between the area of the smallest disk Disk$(F)$ containing $F$ (outer disk) and the area of $F$ is at most  $q\in\mathbb{R}$.
If  $\Cscr$ consists of squares,  then $q=\frac\pi 2$, so  Pach's bound is $\chi(\Cscr) \leq 9\frac\pi 2 \Delta(\Cscr) \approx 14.14 \Delta(\Cscr)$.
This can be essentially improved: 

 \begin{theorem}\label{thm:colouring} 
 If $\Cscr$ is a family of squares and $\Delta(\Cscr) \geq 2$,  $\chi(\Cscr)\le 9(\Delta(\Cscr) - 1) $.
 Moreover, if the squares have equal size,
 $\chi(\Cscr)\le 6\Delta(\Cscr)$.
\end{theorem}

Other results on the chromatic number of families of convex sets can be found in Table \ref{tab:chi_bounds}. Recall that for any family $\Fscr$, $\Delta(\Fscr) \leq \omega(\Fscr)$. 

\begin{table}
\centering
\begin{tabular}{c|c|c|c|c}
 & \textit{convex set} & \textit{centrally symmetric} & \textit{disk} & \textit{square} \\ \hline
\textit{translation} &  $\chi \leq 3\omega - 2$   &  $\chi \leq 3\omega - 2$   & $\chi \leq 3\omega - 2$  & $\chi \leq 2\omega - 1$ \ \\ 
 &   \cite{2004_Kim}  &  \cite{2004_Kim}   &  \cite{2003_Perepelitsa} & \cite{2003_Perepelitsa} \\ \hline
\textit{translation+} & $\chi \leq 6\omega - 6$  &  $\chi \leq 6\omega - 6$  & $\chi \leq 6\omega - 6$  & $\chi \leq 4\omega - 3$  \\ 
\textit{homothety} &  \cite{2004_Kim} &   \cite{2004_Kim}  & \cite{2004_Kim} &  \cite{2006_Ahlswede}  \\ \hline
\textit{translation+}&  $\chi \leq 9 q \Delta$ &  $\chi \leq 9 q \Delta$ & $\chi \leq 3\omega - 2$ & $\chi \leq 6\Delta$ \\ 
\textit{rotation} & \cite{1980_Pach} & \cite{1980_Pach} & \cite{2003_Perepelitsa} & Thm. \ref{thm:colouring} \\ \hline
\textit{translation+} &  $\chi \leq 9 q \Delta$  & $\chi \leq 9 q \Delta$ & $\chi \leq 6\omega - 6$ & $\chi \leq 9(\Delta-1)$  \\
\textit{homothety+} & \cite{1980_Pach} & \cite{1980_Pach} & \cite{2004_Kim} & Thm. \ref{thm:colouring} \\
\textit{rotation} &  &  &  & \\ 
\end{tabular}
\caption{$\chi/\omega$ and $\chi/\Delta$ bounds  for   translations, rotations, or homotheties of a convex set.} 	
\label{tab:chi_bounds}
\end{table}

We do not know about non-trivial lower bounds for colouring squares. The intersection graph of unit squares may be a $C_5$ with chromatic number $\chi=3$ and clique number $\omega=2$.
 However, the $3/2$ lower bound cannot be easily kept for higher values of $\omega$: the best-known bound arises by choosing $\omega$ to be divisible by $4$,  taking each square of the $C_5$ example $\omega/2$ times. In terms of the intersection graph, this is a replication of each vertex  $\omega/2$ times, and an optimal colouring is provided then by taking each of the five maximal stable sets of $C_5$ $\omega/4$ times, as colour classes, showing $\chi=\frac54\omega$ (see \cite{2012_Dumitrescu}).
This seems to be the best example known for squares of varying sizes and not necessarily axis-parallel as well. 

The paper is organized as follows.
Section~\ref{sec:cover} introduces the main techniques of the paper. First, it formalizes the greedy argument that connects hitting degeneracy with the hitting number and presents a proof of Pach's bound, which also serves as an initial example of the \quotes{local homothety operation.}  Then, a simple and well-known relation is explored between hitting a family of geometric objects and covering a set of points by such a family.  This relation allows us to show some first bounds on the $\tau/\nu$ ratio, in particular, a bound for the hitting number (Theorem~\ref{thm:hit_convexsets}) similar to Pach's bound for the chromatic number, and also to point at the incompleteness of this method:
while an appropriately defined covering is always sufficient for defining a hitting set, it is necessary only in the axis-parallel case.
In Section \ref{sec:holes}, we develop the main tools that essentially improve the \quotes{covering method,} enabling us to turn even some partial coverings into hitting sets, enhancing our bounds. 
In Section~\ref{sec:proofs}, we use the tools from the previous sections to prove the upper and lower bounds for hitting (Theorem~\ref{thm:hitting}, Theorem~\ref{thm:lower}) and the upper bound for the chromatic number (Theorem~\ref{thm:colouring}).  
Section~\ref{sec:end}  concludes the paper with a collection of open questions.

\section{Initial framework and basic tools}\label{sec:cover}

In this section, we first set the framework of the paper and then gradually build the techniques towards our main results:

Section~\ref{sub:hitting_degeneracy} works out necessary details for the use of hitting-degeneracy and shows a first example of the local homothety operation by shortly presenting  Pach's proof for his general bound for the chromatic number~\cite{1980_Pach}. Section~\ref{sub:hitting_covering} explains the relationship between the hitting sets of a family of unit balls and sets covering their centres. This transition from hitting to covering offers a new perspective on the problem and allows us to show a first bound on the hitting number of squares that will be improved later. 
Finally, Section~\ref{sub:hitting_general_bound} applies these techniques for bounding the hitting number of convex sets.

\subsection{Hitting degeneracy and local homothety}\label{sub:hitting_degeneracy}

We use a tool similar to $D$-degeneracy to bound the hitting number of squares.
The following well-known lemma \cite{2006_Kim, 2011_Dumitrescu} formulates the simple and standard induction step converting constant size hitting sets of neighbourhoods into a linear bound between $\tau$ and $\nu$.

\begin{lemma} \label{lem::induction_idea}
	Let $\Fscr$ be a hitting-$t$-degenerate family of sets. Then $\tau(\Fscr) \leq  t\nu(\Fscr)$. Moreover, if $\tau(\Fscr')\le t_0$ for any $\Fscr'\subseteq \Fscr$ satisfying $\nu(\Fscr') = 1$, then
	$\tau(\Fscr) \leq t_0 + t(\nu(\Fscr) - 1).$ 
\end{lemma}
\begin{proof}
We proceed with the proof of the second statement by induction on $\nu(\Fscr)$ since then the first statement follows by substituting $t_0:=t$. 
If $\nu(\Fscr) = 1$, then $\tau(\Fscr)\leq t_0$  by the condition.
Assume now that $\nu(\Fscr) \geq 2$ and that the statements are true for any subfamily $\Fscr'$ of $\Fscr$ having $\nu(\Fscr') < \nu(\Fscr)$.
Let $F \in \Fscr$ be given by the condition, that is,  $N[F]$ can be hit by $t$ points.
Any packing $\mathcal{I}'$ in $\Fscr - N[F]$ has size at most $\nu(\Fscr)-1$ since $\mathcal{I}' \cup \{F\}$ is a packing  in $\Fscr$.
Then, since the condition is still satisfied by $\Fscr - N[F]$, by the induction hypothesis,
\begin{align*}
	\tau(\Fscr) &\leq \tau(\Fscr - N[F]) + \tau(N[F]) \\ &\leq t_0 + t(\nu(\Fscr - N[F]) -1)+t\\
	&= t_0 + t(\nu(\Fscr)  -2)+t \\&= t_0 + t(\nu(\Fscr) -1).
\end{align*}
\end{proof}

The following lemma presents a straightforward connection between degeneracy and hitting degeneracy.
Together with Lemma \ref{lem:hitting_neighbours}, it will allow us to derive an upper bound on the chromatic number of squares.

\begin{lemma}\label{lem:degimplication}
Any hitting-$t$-degenerate family of sets $\mathcal{F}$ is also $(t\Delta(\mathcal{F}) - 1)$-degenerate.
\end{lemma}
\begin{proof}
Let $\mathcal{F}'$ be a subfamily of $\mathcal{F}$. Since, $\Fscr$ is hitting-$t$-degenerate, there is a set $F' \in \Fscr'$ such that $\tau(N[F']) \leq t$.
We show that $F'$ has at most  $t\Delta(\Fscr) - 1$  neighbours. Each of the points in a hitting set of $N[F']$ is contained in at most $\Delta(\Fscr') \leq \Delta(\Fscr)$ sets. 
Hence, $|N[F']| \leq \tau(N[F']) \Delta(\Fscr) \leq t\Delta(\Fscr)$ and so $|N(F')| \leq t\Delta(\Fscr) - 1$.
\end{proof}

Now, we present Pach's result, the proof of which details the \emph{local homothety operation}. 
\begin{theorem}[Pach, 1980 \cite{1980_Pach}]\label{thm:Pach}

Let $q\in \mathbb{R}_+$ and $\Fscr$ be a family of convex sets in the plane.
If for each $F\in\Fscr$  the ratio between the area of the smallest disk $\disk(F)$ containing $F$ and the area of $F$ is at most  $q$, then  $\chi(\Fscr) \leq 9q \Delta(\Fscr)$.

\end{theorem}
\begin{proof}
The chromatic number of any $D$-degenerate family is bounded by $D+1$. Hence, it is enough to show that $\mathcal{F}$ is $(9q \Delta(\Fscr)-1)$-degenerate.
Let $F_0\in\Fscr$ be such that $\disk(F_0)$ has the shortest radius among $F\in\Fscr$ and choose the unity to be the length of this radius.
We check  $|N(F_0)| \leq 9q \Delta(\Fscr)-1$.

For each  $F\in N(F_0)$, define a new convex set $F'$ by picking an arbitrary point $p\in F\cap F_0$ and applying an appropriate homothety with centre $p$ and ratio $\lambda \leq 1$ so that the image of the outer disk of $F$ (which coincides with $\disk(F')$) has radius $1$.
We call this operation {\em local homothety.} 
Let $c_0$ be the centre of  $\disk(F_0)$ and note that for all $F\in N[F_0]$, each outer disk $\disk(F')$, and so $F'$, is contained in the disk $B$ of centre $c_0$, and radius $3$. 

Since, by convexity  $F' \subseteq F$, the local homothety operation did not increase the maximum degree of $\Fscr$. 
Therefore the disk  $B$ is covered by the images $\{F' : F\in N[F_0]\}$ at most $\Delta(\Fscr)$ times.  Moreover, the areas of the images  $\{F' : F\in N[F_0]\}$, each of which is at least $\pi/q$ by the definition of $q$, sum up to less than $9\pi\Delta(\Fscr)$. These two observations yield $|N(F_0)| \leq  |N[F_0]| -1 \leq \frac{9\Delta(\Fscr)\pi}{\pi/q} -1 =9q\Delta(\Fscr) -1$.

\end{proof}

Surprisingly, in both the problem of bounding the chromatic numbers of intersection graphs of geometric objects in terms of their maximum degrees and the problem of bounding their hitting numbers in terms of their packing numbers, it seems that no approach is known that takes into consideration the family of objects in a more global way.
While we continue exploiting the well-known greedy framework of degeneracy and hitting-degeneracy (Lemma~\ref{lem::induction_idea}), often used for bounding the chromatic and hitting numbers of convex sets, we now also frontally face new challenges for bounding the parameters $t$ and $t_0$.
Doing this, our questioning is {\em not} about the mere existence of constant bounds, which are much easier and follow from Pach's generic method (Theorems \ref{thm:Pach} and \ref{thm:hit_convexsets}). It is rather inspired by the spirit of combinatorial optimization, where the goal is to find the best possible constant. 
The endeavour to find the best bound leads to geometry problems that are interesting in their own right and
can be addressed using classical results, such as Thales' Theorem, in novel ways.

\subsection{Relating hitting and covering}\label{sub:hitting_covering} 

In this section, we consider the relation of covering problems to hitting problems for various geometric objects.
These two problems are equivalent for axis-parallel unit squares and the equivalence is a special case of a general statement about unit balls of normed spaces of arbitrary dimension. 

For an arbitrary norm $||\,\,\,||$,  a  {\em ball} is a set of the form 
\begin{equation*}
    B(c,r):=\{x\in\mathbb{R}^n: ||x-c||\le r  \}\, (c\in \mathbb{R}^n, r\in \mathbb{R}_+).
\end{equation*}
The point $c$ is called the {\em centre} of the ball, $r$ is its {\em radius}, and $\overline{B}(c,r)$ is its \emph{boundary}.
A ball centred at $0$, of radius $r$ is a compact convex set, moreover it is {\em centrally symmetric}   that is, for $x\in B$, $-x\in B$. Conversely, it is also true that any centrally symmetric compact convex set with a non-empty interior is the unit ball for a norm.
   
We primarily use the  $l_2$-norm, also called {\em Euclidean-norm}, and the   $l_\infty$-norm, also called {\em max-norm}; the only distance function we use is $\dist(x,y):=||x-y||_2$, omitting the index. 
In two dimensions, a ball $B_2(c,r)$ for the Euclidean norm is called a  {\em disk} of radius $r$ and its boundary $\Bcirc(c,r)$  is called a {\em circle}; similarly, for the max-norm, $B_\infty(c,r)$ and $\overline{B}_\infty(c,r)$  are an axis-parallel square of side $2r$ and its boundary. 
Note the difference between axis-parallel unit squares $B_\infty(c, 1/2)$, i.e., $1\times 1$ squares,  and  unit balls $B_\infty(c,1)$, i.e., $2\times 2$ squares.

Some further notations will be useful: 
if $X$ is a set of points, $\conv(X)$ denotes their convex hull; if $X=\{a,b\}$, we use use the shorter notation $[a,b]:=\conv (X)$. 
For a square $C$, $l(C)$ is the length of a side of $C$, $c(C)$ denotes the centre of $C$, and given a family of squares $\Cscr$, $c(\Cscr) := \{c(C) \ : C \in \Cscr\}$.

Let $V \subseteq \mathbb{R}^2$ and $|| \ ||$ be a norm.
A {\em covering} with respect to $V$ is a set $\Bscr$ of balls of radius $\frac12$ such that $V \subseteq \bigcup_{B \in \Bscr} B$. A set of points $A \subseteq V$ is {\em independent} if for any pair of different points $u,v\in A$, $||u-v||>1$. The minimum size of a covering, the {\em covering number}, is denoted by $\zeta(V)$, and the maximum size of an independent subset of points by $\alpha(V)$. 
Note that $\alpha\le\zeta$ since the distance of any two points $x,y\in B(c, 1/2)$  is at most $1$.
We will mainly use $\zeta$ and $\alpha$ for the Euclidean norm $l_2$ or the max-norm $l_\infty$ by using the indices $2$ or $\infty$.

\begin{proposition} \label{prop:nu_alpha}
For any norm and family $\Bscr$ of balls of radius $\frac12$, $\nu(\Bscr) = \alpha(c(\Bscr))$ and $\tau(\Bscr) = \zeta(c(\Bscr))$.
\end{proposition}
\begin{proof}
The first equality follows from the simple fact that $\Pscr\subseteq\Bscr$   {\em is a packing if and only if for all $p,q \in c(\Pscr)$, $|| p-q || > 1$.} 
For the second equality note that  $H$ is a hitting set of $\Bscr$, if and only if for each $c\in c(\Bscr)$ there exists $h\in H$ such that  $||c-h||\le \frac12$. This means \ that the balls of radius $\frac12$ with centres in $H$ cover $c(\Bscr)$.  
\end{proof}

 For the max-norm, we immediately obtain the following result.

\begin{corollary} \label{cor:nu_alpha}
For any family $\Cscr$ of axis-parallel unit squares, $\nu(\Cscr) = \alpha_\infty(c(\Cscr))$ and $\tau(\Cscr) = \zeta_\infty(c(\Cscr))$.
\end{corollary}

Considering a family of axis-parallel unit squares,  the neighbours of any of its squares can be hit by at most four points (the four vertices), and the ones of a left-most square by at most two points (the two vertices on the right side).
This simple fact can be reformulated and proved as follows:
the centres of all possible axis-parallel unit squares intersecting a unit square $C$ form an axis-parallel square of size $2\times 2$ with centre $c(C)$.
Hence, four unit squares are enough to cover them, and only two suffice if the centre of $C$ is a left-most one (see Figure \ref{fig::FourSquares}).
By Corollary~\ref{cor:nu_alpha}, $\tau( c(N[C]) )=\zeta_\infty( c(N[C]) )$ is at most $4$ and $2$, respectively.
The greedy induction and the local homothety presented in Section~\ref{sub:hitting_degeneracy} can be used to convert these bounds into $\tau\le 4\nu -3$ for axis-parallel squares of arbitrary size and $\tau\le 2\nu - 1$ for axis-parallel unit squares \cite{2006_Ahlswede}. 

\begin{figure}[ht]
	\centering
		\includegraphics[scale =0.6]{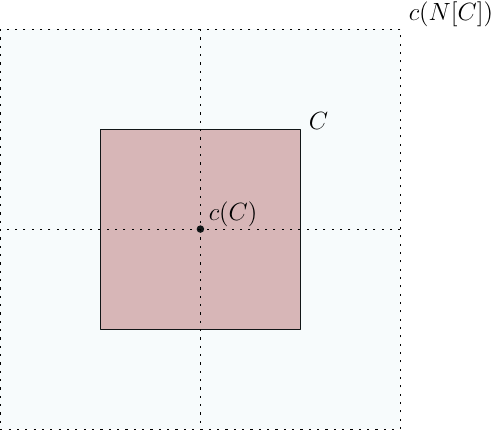}
		\caption{An axis-parallel square and the domain of the centres of its neighbours.}
		\label{fig::FourSquares}
\end{figure}

Corollary~\ref{cor:nu_alpha} is not directly applicable to not necessarily axis-parallel squares. 
We make a detour through other norms to still apply the second part of Proposition~\ref{prop:nu_alpha} at the price of losing a small constant factor.
The inner disk of a unit square $C$ is $B_2(c(C), \frac{1}{2})$  and its outer disk is $B_2(c(C), \frac{\sqrt{2}}{2})$.
\textit{Hitting the inner disks, we also hit the original squares}, and \textit{packing the outer disks, we also pack the correspondent squares}.

We illustrate now how the hitting number of the neighbours of a unit square can be bounded with the help of the second part of  Proposition~\ref{prop:nu_alpha}, establishing hitting-degeneracy.
This bound is weaker than Lemma~\ref{lem:hitting_neighbours},  which will be proved by completing the covering argument used here with the novel methods of Section~\ref{sec:holes}.

\begin{proposition} \label{prop:hit_12}
Let $\Cscr$ be a family of unit squares and $C_0\in\Cscr$. Then $\tau(N[C_0])\le 12$, so $\tau(\Cscr)\le 12\nu(\Cscr)$.  
\end{proposition}
\begin{proof}
By Lemma~\ref{lem::induction_idea}, the second inequality is a straightforward consequence of the first one.
The centres of any unit square intersecting $C_0$ are clearly contained in a square $T$ of size  $\sqrt{2} + 1$. Denote by $\Bscr$ the family of disks we get by replacing each $C\in N[C_0]$ by its inner disk $B\subseteq C$ of radius $\frac12$. According to Proposition~\ref{prop:nu_alpha}, $\tau(N[C_0])\le \tau(\Bscr) = \zeta_2(c(\Bscr))$, in other words, $\tau(N[C_0])$ can be upper bounded by the minimum number of disks of radius $\frac12$ that cover $T$. To prove that $12$ such disks are sufficient, it is enough to give these disks. Since we will prove the better bound $10$ with a more powerful method, we just refer to a result of Nurmela and \"Osterg{\aa}rd. 
In \cite{2000_Nurmela}, they provide for $1\le n\le 30$ the minimum $r_n$ of the equal radii of $n$ disks covering a unit square. Then proportionally, $\frac1{2r_n}$ is the maximum size of a square that can be covered by disks of radius $\frac12$, and $\frac1{2r_{12}}> \sqrt{2} + 1$.
\end{proof}

To summarize, the hitting and covering problems are equivalent for axis-parallel unit squares (Corollary~\ref{cor:nu_alpha}), but they are not if some of the given unit squares are not axis-parallel.
We overcome this difficulty by covering with the inner disks of the unit squares instead of the unit squares themselves and packing the outer disks. In this way, the method can be saved at the price of weaker bounds. The loss can be decreased by showing that the covering does not have to be perfect: ``holes'' with certain metric properties can be \quotes{patched.}  
The main progress of our work is reached by three observations stating that the centres of the covering disks actually hit \emph{more} squares than those with centres in the union of the disks. They also hit the squares whose centres are in the holes if these uncovered territories are small enough.
For example, Figure \ref{fig:3_balls} 
shows three light (blue\footnote{The text is meant to be understandable without the colours as well.}) disks of diameter $1$, whose centres hit each unit square having its centre in the light (blue) region, since it even hits their inner disks by Corollary~\ref{cor:nu_alpha}; however, even if the centre is not in the light (blue) region, but in the dark (blue) \quotes{hole,} as for the represented unit square, the square is hit by one of the three centres. So the dark (blue) hole can be ``patched'' (Lemma ~\ref{lem:inside}). In Sections \ref{sec:holes} and \ref{sec:proofs}, the idea of allowing holes in the cover and covering them with ``patches'' is further explored and used to improve the bound offered by Proposition~\ref{prop:hit_12}.

\begin{figure}[ht]
	\centering
		\includegraphics[scale =0.7]{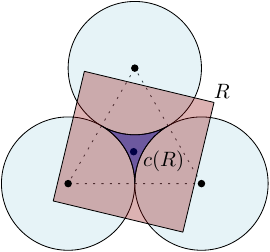}
		\caption{A ``hole'' not covered by any of three disks, but ``patched'': the three vertices of the triangle hit all unit squares having their centre in the dark (blue) hole as well.}
		\label{fig:3_balls}
\end{figure}

\subsection{A bound on the hitting number of convex sets}\label{sub:hitting_general_bound}
In this section, we show how to apply the tools introduced in Sections~\ref{sub:hitting_degeneracy} and~\ref{sub:hitting_covering} to convex sets in general. Focusing exclusively on how the covering tool can be used will help us construct more complex geometric arguments by mixing several tools.
At the same time, we see that such a brute force application of the covering tool already shows a constant bound for  $\tau/\nu$ for the translates of rotated homothetic copies of a fixed convex set.
The rest of the paper will then enrich this framework with novel geometric arguments, providing essential improvements for squares.

For a compact convex set $K$, we can define the {\em slimness} of $K$, denoted by $\ar_{||\,||}(K)$, as the ratio of the radius of the smallest ball (in norm $||\,||$) containing $K$ and the largest ball contained in $K$. If the latter is $0$, the ratio is $\infty$.
For the most commonly used Euclidean and max norms: $\ar_2(K)$ denotes the ratio of the radius of the smallest disk containing $K$ and the largest disk contained in $K$; similarly, $\ar_{\infty}(K)$ is the ratio of the sides of the smallest axis-parallel square containing $K$ and the largest axis-parallel square contained in $K$.
Observe that $\ar_2(K)$ is invariant under translations, rotations and homothety.
Moreover, since a disk of radius $r$ is contained in a square of size $2r$ and contained a square of size $\sqrt{2}r$,  the parameter $\ar_2(K)$ and $\ar_{\infty}(K)$ are within a factor of $\sqrt{2}$ from each other.
Pach's coefficient $q(K)$ for a convex set $K$ is the quotient of the {\em areas}  of the outer disk of $K$, and $K$ itself, implying $q(K)\le\ar_2(K)^2$. 
The slimness is a parameter that has been often used when studying the packing and hitting number (see, for example, \cite{2006_Ahlswede} and \cite{2003_Chan}).

The proof of the following theorem is based on the local homothety operation, applied here to the $\tau/\nu$ ratio. 

\begin{theorem}\label{thm:hit_convexsets}
Let $\Fscr$ be a family of convex sets in the plane. Then
\begin{equation*}
(i) \ \tau(\Fscr) \leq 9 \ar_{\infty}(\Fscr) ^2 \nu(\Fscr).
\end{equation*}
Moreover, if the sets are centrally symmetric,
    $(ii) \ \tau(\Fscr) \leq 4 \lceil  \ar_{\infty}(\Fscr)    \rceil^2 \nu(\Fscr)$,
and
    $(iii) \ \tau(\Fscr) \leq 2 \lceil  \ar_{\infty}(\Fscr)    \rceil^2  \nu(\Fscr$),
if the inner squares are all of the same size.
\end{theorem} 


\begin{proof} 
We proceed by using local homotheties to deduce hitting-degeneracy (as Pach~\cite{1980_Pach} did by degeneracy for the chromatic number, see the proof of Theorem \ref{thm:Pach}), and then, applying Lemma~\ref{lem::induction_idea}.  In the proof, we use the max-norm $\ar_{\infty}(\Fscr) = \ar_{\infty}$, so a ball of radius $r\in\mathbb{R}$ is an axis-parallel square of side length $2r$. All the squares in this proof are axis-parallel so we omit this specification.   

Let $F_0\in\Fscr$ be a convex set with the smallest inner square and assume for simplicity that the side length of its inner square is $1$, and $F_0 \subseteq B_{\infty}(c_0,\frac{\ar_{\infty}}{2})$ for a $c_0\in\mathbb{R}^2$.
Apply local homotheties: for each $F \in N(F_0)$ consider a point $f \in F_0 \cap F$ and take a homothetic copy $F'$ of $F$ with centre  $f$ and ratio $\lambda \leq 1$ so that the image of the inner square of $F$ (which is the inner square of $F'$) is a unit square.  Define $F_0':=F_0$. Since every $F$ is convex, we have that:
\begin{align*}
    &\textrm{(a)} \ F' \subseteq F; \\
    &\textrm{(b)} \ f\in F_0 \cap F' \textrm{, in particular, } F_0 \cap F'\ne\emptyset; \\
    &\textrm{(c)} \ \textrm{There is a square of side length } \ar_{\infty} \textrm{containing } F'.
\end{align*}

By (a), we can hit the sets in $N[F_0]$ by hitting  $N':=\{F': F\in N[F]\}$, which can be achieved, in turn, by hitting their inner squares. Denote the family of these inner squares by $S'$.  Properties (b) and (c) imply then that all sets in $N'$,  are contained in $B_\infty(c_0,\frac32 \ar_{\infty})$, and therefore the centres of the inner squares of $S'$, all of side length $1$, are in  $B_\infty(c_0,\frac32(\ar_{\infty} - 1))$ (Figure \ref{fig:large_outer_disk}).
By Corollary~\ref{cor:nu_alpha}, then hitting {\em all possible inner squares} is equivalent to covering the square $B_\infty(c_0,\frac32(\ar_{\infty} - 1))$ by unit squares.
Hence, 
\begin{equation*}
    \tau(N[F_0]) \leq  \tau(N') \leq \tau(S') \leq \zeta(B_\infty(p,\frac32(\ar_{\infty} - 1))).
\end{equation*}
We immediately get
$\zeta(B_\infty(c_0,\frac32(\ar_{\infty} - 1))) \leq  \lceil  3 (\ar_{\infty} -1)    \rceil ^2 \leq 9 \ar_{\infty} ^2.$
Applying this to an arbitrary subfamily of $\Fscr$, we see that it is hitting-$9 \ar_{\infty} ^2$-degenerate, so  $(i)$ follows from Lemma~\ref{lem::induction_idea}. 

\begin{figure}[ht]
    \centering
    \includegraphics[scale=0.6]{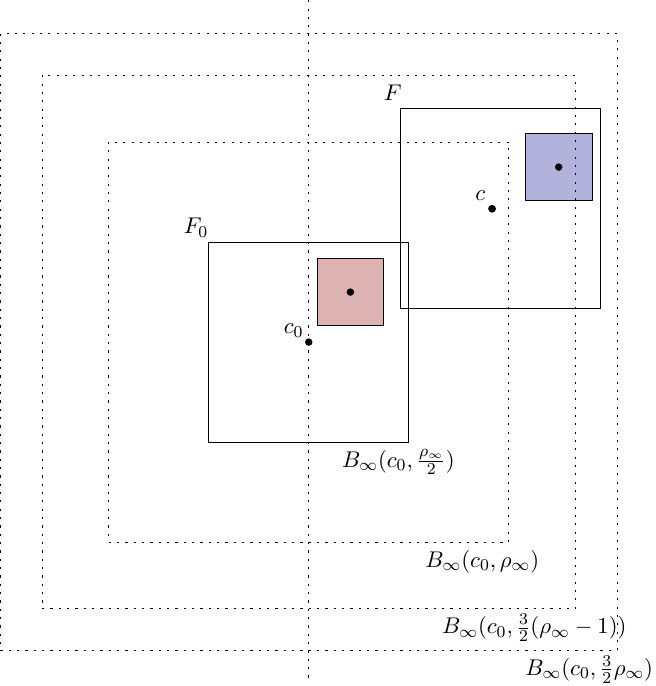}
    \caption{The outer and inner squares of two intersecting convex sets with slimness $\ar_{\infty}$ and inner unit squares.}
    \label{fig:large_outer_disk}
\end{figure}
The proof of $(ii)$ follows with the only difference that it exploits the fact that the centres of the inner and outer squares coincide, and the centres of outer squares of sets in $N'$ are all in $B_\infty(c_0,\ar_{\infty})$, the square of centre $c_0$ and side length $2\ar_{\infty}$.  

If, in addition, the inner squares of the centrally symmetric sets in $\Fscr$ have all the same size,  local homothety is not needed anymore, leaving
 us free to choose $F_0\in\Fscr$ to have in addition an inner square with a left-most centre. Now the centres of outer squares are contained in one half of $B_\infty(c_0,\ar_{\infty})$, denote it by $M$. We get
$ \tau(N[F_0]) \leq \zeta(M) \leq \lceil   \ar_{\infty}    \rceil \lceil 2 \ar_{\infty}    \rceil \leq 2\lceil   \ar_{\infty}    \rceil^2$.
\end{proof}

In order to bound the hitting number of translates of rotated homothetic copies of a fixed convex set $K$, it is unfortunately not sufficient to substitute its slimness to  Theorem \ref{thm:hit_convexsets}, since $\rho_{\infty}$ is not invariant under rotation.  However,  $\rho_{2}$ is, and using this, the theorem can be applied. Since $\rho_{\infty}(K) \le \sqrt{2}\ar_2(K)$, this only adds a factor of $2$ in the bounds.

\begin{corollary}\label{cor:trans_rot_homothetic}
Let $\Fscr$ be a family of translates of rotated homothetic copies of a fixed convex set $K$ in the plane. Then
\begin{equation*}
    \tau(\Fscr) \leq 18 \ar_{2}(K) ^2 \nu(\Fscr).
\end{equation*}
The upper bound decreases to $8 \lceil \ar_{2}(K) \rceil ^2 \nu(\Fscr)$ when $K$ is centrally symmetric, and further to $4 \lceil \ar_{2}(K) \rceil ^2 \nu(\Fscr)$ if $K$ is centrally symmetric and $\Fscr$ contains exclusively translates of rotated copies of $K$.
\end{corollary}

These estimates are, of course, rough but do satisfy the modest goal of showing how local homothety applies to exploiting slimness and how to take advantage of particularities like central symmetry to sharpen the results.
For squares and unit squares, we will be more meticulous, focusing on obtaining the best bound we can.

\section{Filling holes}\label{sec:holes}
This section finds ``patches'' for ``holes'' uncovered by disks.
Patches cover ``for free," i.e., without using more disks for the covering.
A first kind of hole,  shown by Figure \ref{fig:3_balls}, is discussed and patched in Section~\ref{sub:inside}. 
Section~\ref{sub:outside} further develops this technique by finding another set of patches using Thales' celebrated theorem. 
Section~\ref{sub:polygon} describes how to combine the previous two covering tools to hit squares intersecting a convex polygon.
Finally, Section~\ref{sub:triangle} completes the picture by showing one more patch for filling holes.  

The initial ``covering disks'' (Section \ref{sub:hitting_covering}) strengthened by these patches allow us to prove our best hitting/packing ratio (Section~\ref{sec:proofs}).
 
\subsection{Filling  holes inside triangles}\label{sub:inside}

First, we prove the assertion anticipated by Figure~\ref{fig:3_balls}. 

\begin{lemma}[Triangular patch]\label{lem:inside} Let $a,b,c \in \mathbb{R}^2$ be  three points at distance at most $1$ from one another.  Then any square $C$ of sides at least $1$ and centre $c(C)$ in $T:=\conv(\{a, b, c\})$ contains at least one point from $\{a,b,c\}$.
\end{lemma}
\begin{proof}
By the condition,  $c(C)\in K:=T\cap C$.  We prove that {\em the polygon   $K$ has a vertex which is a vertex of $T$.}  

Suppose for a contradiction that this is not true. Then the vertices of $K$ are vertices of $C$ or intersections of a side of  $C$ and of a side of $T$, so {\em the vertices of $K$ lie on the sides of $C$, and actually all of them lie on two intersecting sides $[x,y]$,  $[y,z]$ of $C$} since the distance of any pair of vertices on two distinct parallel sides of $C$ is at least $1$, while the distance of any pair of points in $T\setminus \{a,b,c\}$ is strictly smaller than $1$.
However, among the convex hulls of pairs of points in $[x,y]\cup [y,z]$,  only $[x,z]$ contains $c(C)$. Since $c(C)\in K$, we have  $x, z \in K\subseteq T$. But $\dist(x,z)>1$, a contradiction.
\end{proof}

\subsection{Filling holes outside separating lines}\label{sub:outside} 

The following equivalence is an immediate consequence of Thales' Theorem.

\begin{proposition}\label{prop:Thales}
Let $a,b$ be two points on a line $L\subseteq\mathbb{R}^2$, $q$ the midpoint of the segment $[a,b]$, and $c$ a point in $\mathbb{R}^2 \setminus L$.
Then there exists a right-angled triangle with a closed subsegment of the open interval $(a,b)$ as hypotenuse and $c$ as the third vertex if and only if $\dist (c, q) < \dist (q,a)$, i.e.,  $c$ is in the open disk with centre $q$ and radius $[q,a]$.
\end{proposition}

We say that a line {\em separates} two sets if these sets are in two different open half-planes bordered by the line. 
The following lemma is still essentially Thales's theorem, reformulating and completing the equivalence of  Proposition \ref{prop:Thales} into a form  comfortable  for hitting sets:

\begin{lemma}\label{lem:outside} 
Let $a,b$ be two points on a line $L\subseteq\mathbb{R}^2$, $\dist (a,b) \leq 1$, let $q$ be the midpoint of the segment $[a,b]$, and $c\in\mathbb{R}^2\setminus L$.
If $\dist(q, c)\ge \dist(q, a)$, then each unit square $S$,  $c\in S$ so that $L$ separates $c(S)$ from $c$, and  $S\cap [a,b] \ne\emptyset$ contains either $a$ or $b$ (Figure \ref{fig:Lemma2}).
\end{lemma}

\begin{figure}[ht]
	\centering
		\includegraphics[scale =0.45]{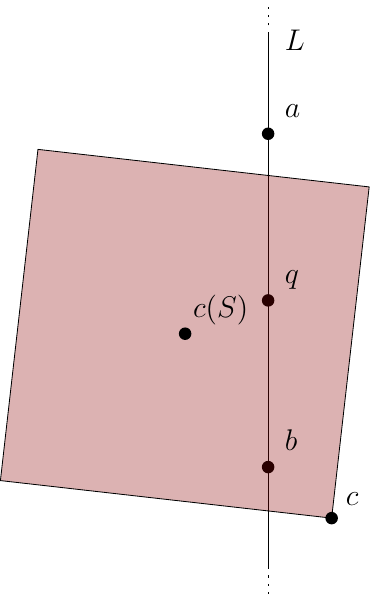}
		\caption{Illustration of Lemma~\ref{lem:outside}.}
		\label{fig:Lemma2}
\end{figure}

\begin{proof} 
Assume for a contradiction that  $\dist (q, c) \geq \dist(q, a)$, but there exists a unit square $S$ whose centre is separated from $c$ by $L$, which meets $[a,b]$,  contains $c$, but neither $a$ nor $b$. By the separation and the convexity of $S$, $S\cap [a,b]$ is a non-empty, closed subsegment of the open interval $(a,b)$, and since $\dist(a,b)\le 1$, no two parallel sides of $S$ can meet $[a,b]$. 

Therefore the half-plane containing $c$ intersects $S$ in a right-angled triangle $T' := \conv(a',b',c')$,  with hypotenuse $[a',b'] \subset [a,b]$, and right angle in $c'$.
Applying the if part of Proposition~\ref{prop:Thales} to the point $c'$ and the triangle $T'$, we have that $c'$, and consequently $T'$, is contained in the disk with centre $q$ and radius smaller than $[q,a]$. Since  $c \in T'$, $\dist(q,c)< \dist(q,a)$ contradicting the indirect assumption.
\end{proof}

\subsection{Hitting squares using polygons}\label{sub:polygon}

Combining Lemma~\ref{lem:inside} and Lemma~\ref{lem:outside}, we get a sufficient condition for hitting each unit square intersecting a compact set inside a convex polygon and far enough from its vertices. This tool is presented  in the following theorem:
\begin{theorem}\label{thm:in_out_convex_hull}
	Let $P \subseteq \mathbb{R}^2$ be a convex polygon with vertices denoted by $p_1, \ ..., p_k$ and let $p_0 \in P$.
    For each $1\leq i \leq k$, denote by $q_i$  the midpoint of the side  $[p_i, p_{i+1}]$ of $P$, where $p_{k+1} := p_1$. 
    Assume that:
	\begin{itemize}
		\item[$(i)$] $\dist(p_0, p_i) \leq 1 \textrm{, for any } 1 \leq i \leq k,$
		\item[$(ii)$] $\dist(p_i, p_{i+1}) \leq 1 \textrm{, for any } 1 \leq i \leq k.$
	\end{itemize}
    Moreover, let $C\subseteq P$ be a closed set. Then any unit square $S$  intersecting $C$ and satisfying the condition
    \begin{itemize}
        \item[$(iii)$]
        $\dist(p_i, q_i) \leq \dist(q_i, C)$, for each $1\le i \le k$ for which  $c(S)$ 
        is separated from $C$ by the line $L_i\supset [p_i, p_{i+1}]$,
    \end{itemize}
    is hit by at least one of the points in $\{p_0,p_1, \ ..., p_k\}$.
\end{theorem}
\begin{figure}
    \centering
    \includegraphics[scale=0.6]{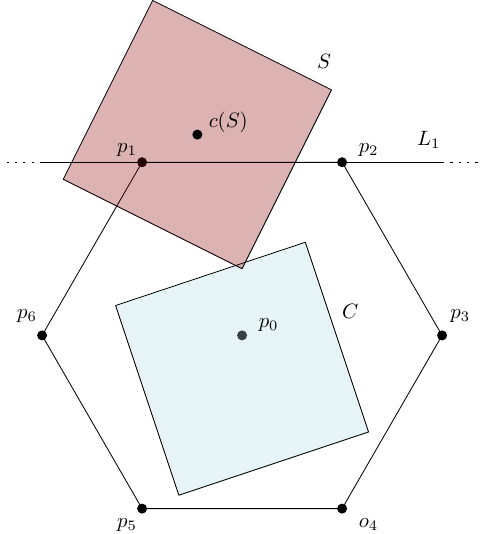}
    \caption{Illustration of Theorem~\ref{thm:in_out_convex_hull}; condition $(iii)$ is required for $i=1$. 
}   \label{fig:enter-label}
\end{figure}
\begin{proof}
First, suppose $c(S)\in P$, and for each $1 \leq i \leq k$,  let $T_i:= \conv(\{p_0,p_i,p_{i+1}\})$.
Note that $P=\bigcup_{i=1}^kT_i$.  Therefore there exists $i\in\{1,\ldots, k\}$ such that $c(S)\in T_i$, and by the conditions $(i)$ and $(ii)$ we can apply Lemma \ref{lem:inside} to $T_i$ with the result that at least one of $a:=p_0$, $b:=p_i$, and $c:=p_{i+1}$ hits $S$.

Second, assume $c(S)\notin P$, and let $c \in S\cap C$. Since $c\in P$, there exists a side of the polygon $P$, say $[p_i,p_{i+1}]$ whose line $L_i$ separates $c(S)$ and $c$. 
Now according to $(iii)$ the condition of  Lemma~\ref{lem:outside} is satisfied for $c$, $a:=p_i$, $b:=p_{i+1}$, so either $p_i$ or $p_{i+1}$ hits $S$. 
\end{proof}

\subsection{Patch using the triangle inequality}\label{sub:triangle}
The following lemma strengthens Lemma~\ref{lem:outside}. Together with Lemma \ref{lem:inside}, it will be used in some situations to fill in the holes left by the simple reformulations by covers of Proposition~\ref{prop:nu_alpha}.

\begin{lemma}[Circular patch]\label{lem:third_ball}
Let $a,b\in \mathbb{R}^2$, $\sqrt{2}-1\le \dist(a,b)\le 1$, let $q$ be the midpoint of the segment $[a,b]$, and $d:=\dist(q,a)$.
Then any square of side at least 1 and centre in $B_2(q,\frac{\sqrt{2}}{2} - d)$ contains either $a$ or $b$   (Figure~\ref{fig:Lemma3}~(b)).
\end{lemma}
\begin{proof}
Assume for a contradiction that there exists a unit square $S$ with centre $c(S)$ in $B_2(q,\frac{\sqrt{2}}{2} - d)$, but containing neither $a$, nor $b$.  Since $S$ contains a disk of radius $\frac{1}{2}$ and, $\dist(q,c(S))\leq \frac{\sqrt{2}}{2}-d \leq \frac{1}{2}$, we have $q \in S$ and so $S\cap [a,b] \ne\emptyset$.

Let $L$ be the line containing $[a,b]$ and let $c$ be a vertex of $S$ separated from $c(S)$ by $L$. By Lemma~\ref{lem:outside} (Thales' theorem), $\dist(q,c) < \dist(q,a) = d$.

Now consider the triangle defined by $c(S)$, $q$ and $c$ (Figure \ref{fig:Lemma3}(a)). By the triangle inequality, we have
\begin{equation*}
    \frac{\sqrt{2}}{2} \leq \dist(c(S),c) \leq \dist(c(S),q) + \dist(q,c) < \left ( \frac{\sqrt{2}}{2} - d \right ) + d = \frac{\sqrt{2}}{2}.
\end{equation*}
This contradiction concludes the proof of the lemma.
\end{proof}

\begin{figure}[ht]
	\centering
	\includegraphics[scale =0.6]{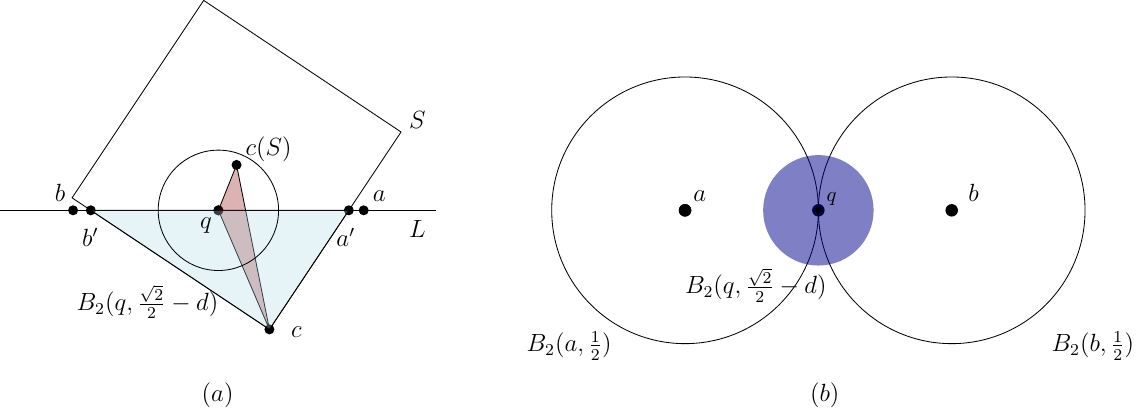}
	\caption{Figure $(a)$ illustrates the setup of Lemma \ref{lem:third_ball}, while the small (blue) disk of Figure $(b)$ shows the  provided ``patch.''}
	\label{fig:Lemma3}
\end{figure}

\section{Deducing the new bounds for squares}\label{sec:proofs}

In this section, we prove the main results of the paper:
In Section \ref{sub:proofs}, we prove the upper bounds for the $\tau / \nu$ ratio of unit squares and squares of different sizes (Theorem~\ref{thm:hitting}). This is the most technical part of the paper, and it requires the tools developed in Sections~\ref{sec:cover} and \ref{sec:holes}.
In Section \ref{sub:lower}, we provide examples showing the lower bounds for the $\tau / \nu$ ratio (Theorem \ref{thm:lower}).
Finally, in Section \ref{sub:colouring}, we prove the upper bounds for the chromatic number of squares (Theorem~\ref{thm:colouring}): the proof for unit squares uses  Lemma \ref{lem:hitting_neighbours}, while the one for squares is inspired by an averaging argument used for instance in \cite{1960_Asplund} and \cite{2021_Chalermsook}.


\subsection{Hitting squares}\label{sub:proofs}

In this section, we first prove the promised upper bounds on the hitting number: Lemma~\ref{lem:hitting_neighbours} leading to Theorem~\ref{thm:hitting}. 

\subsubsection*{Proof of the first part of Lemma~\ref{lem:hitting_neighbours}.}
Let $C$ be an arbitrary unit square in $\mathbb{R}^2$, and suppose that the origin is in $c(C)$ and the horizontal and vertical axes are parallel to the sides of $C$. We want to present $10$ points and apply  Theorem~\ref{thm:in_out_convex_hull} to show that they hit all possible neighbours of $C$. 
For satisfying conditions $(i)$ and $(ii)$ we first find a set of $9$ points 
on  the circle $\Bcirc(p_0,1)$ of centre $p_0:=c(C)$ and radius $1$.  As a first trial, we can choose these to form a regular $9$-gon.  
However, the relation of a regular $9$-gon to $C$ is not regular, so to satisfy condition $(iii)$ for each possible $S \in N(C)$, slight modifications of the  $9$-gon are necessary.
Since the sides of the regular $9$-gon $P = \conv(\{p_1, \ldots, p_9\})$ are significantly smaller than $1$, we have a margin to move the vertices of $P$ on the circle $\Bcirc(p_0,1)$ while preserving conditions $(i)$ and $(ii)$. 
We proceed as follows:

Move two neighbouring vertices closer together when their midpoint $q$ is too close to $C$ to satisfy $(iii)$ (happening when the closest point of $C$ is a vertex of $P$) while move the two vertices of a side away from one another when the midpoint $q$ of the side is relatively far from $C$.   The margin is sufficiently large to easily get points satisfying $(iii)$ without worrying about rounding errors. 

\begin{figure}[ht]
	\centering
	\includegraphics[scale =0.6]{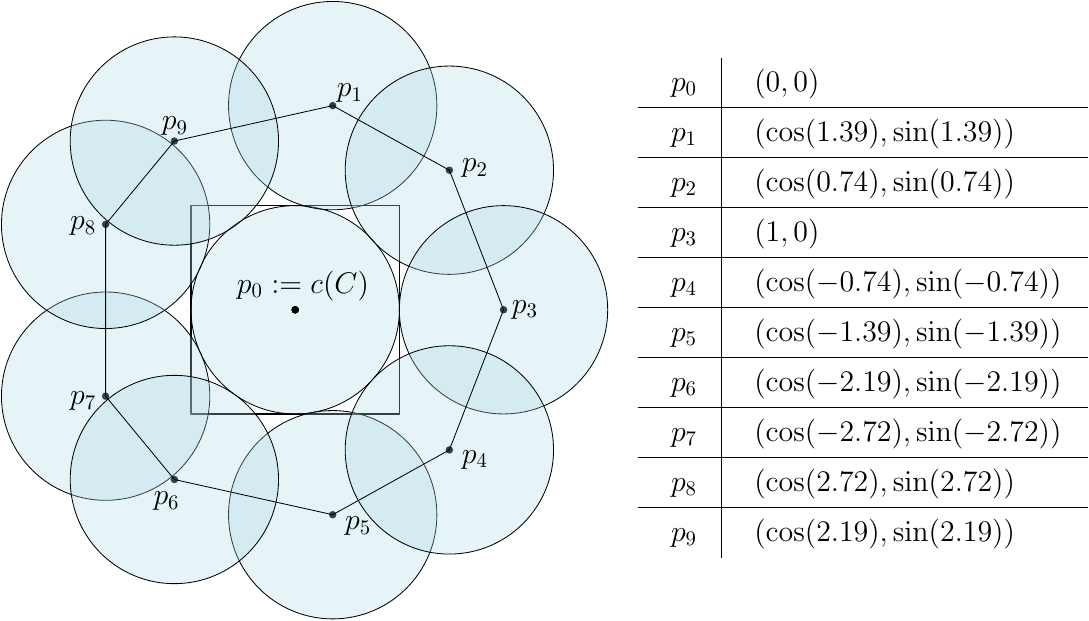}
	\caption{A unit square and the $10$ points hitting its neighbours. Any unit square of centre $p$  with $\dist(p,p_i)\leq1$, i.e., in the light (blue) disks,  is hit by $p_i$, for all $i \in \{0, \ldots, 9\}$. Theorem \ref{thm:in_out_convex_hull} patches the remaining white zones inside $P$.}
	\label{fig:Polygon_9}
\end{figure}

The coordinates of the points we found are given in Figure \ref{fig:Polygon_9}. They obviously satisfy $(i)$, and we can also see without computation that they satisfy  $(ii)$: indeed, it is sufficient to check that the angles between neighbouring vertices of $P$ are not larger than $\pi/3$. They are all smaller than $1<\pi/3$, the largest of them is the angle $p_7p_0p_8$ equal to $0.84$.

In Figure \ref{fig:Check_Prop_iii}, we verify $(iii)$: for all $i \in \{1, ..., 9\}$, 
$\dist(p_i,q_i) < \dist(C,q_i)$ (table of the figure), or, equivalently, $B_2(q_i, \dist(p_i,q_i)) \cap C = \emptyset$ (drawing of the figure), where $q_i$ is the midpoint of the segment $[p_i,p_{i+1}]$ (or $[p_9,p_1]$, if $i=9$).

\begin{figure}[ht]
    \centering
    \includegraphics[scale=0.6]{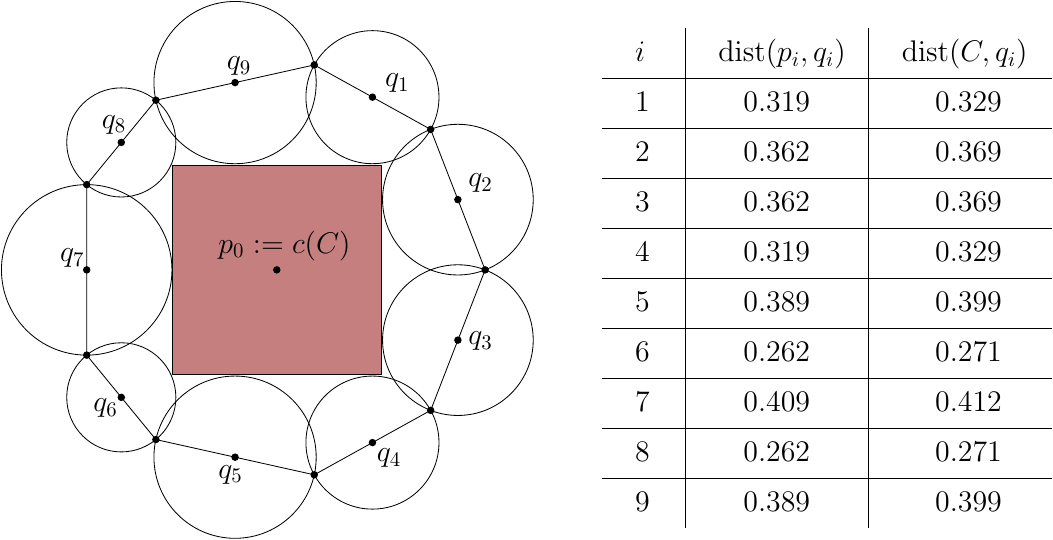}
    \caption{Checking condition $(iii)$ for $1\leq i \leq 9$. The values in the table are rounded to three decimals.}
    \label{fig:Check_Prop_iii}
\end{figure}
Hence, by Theorem~\ref{thm:in_out_convex_hull}, for any possible unit square $S$  intersecting  $C$, we get that $S$ is hit by at least one of the points in $\{p_0, p_1, $\ldots$ , p_9 \}$.

\subsubsection*{Proof of the second part of Lemma~\ref{lem:hitting_neighbours}.}

Suppose the squares of the family $\Cscr$ (consisting of unit squares) are given in $\mathbb{R}^2$,  and that the origin is the left-most centre of a square, denote   $C$ such a square,  $O=c(C)$.   The sides of $C$ are not necessarily parallel to the axes, causing some complications.
The following claim describes a half-disk containing all centres of  possible neighbours of $C$:

\smallskip
\noindent{\bf Claim 1}: All centres of squares in $N[C]$ are in  $Q:= \{(x,y)\in B_2(O, \sqrt{2}): x\ge 0 \}$.\smallskip
\begin{cpf}
Let $S \in N(C)$ be a square of centre $c(S)$. The distance between the centres of two intersecting unit squares is at most $\sqrt{2}$, so $c(S) \in B_2(O, \sqrt{2})$. Moreover, by the choice of $C$,  $c(S)$ is on the right of the vertical axis, finishing the proof of Claim~1.
\end{cpf}

Unfortunately, we cannot immediately rely on some half of the $10$ hitting points of the first part of the proof because a square having its centre to the right of the vertical axis through $p_0$ may have its unique hitting point in the left half-plane.

Let  $v$ be the vertex of $C$ in the non-negative quadrant $\{(x,y)  : x > 0, y \geq 0\}$. We call the angle of $Ov$ with the horizontal axis, the {\em angle} of $C$, and denote it $\angle(C)$. The range of angles to be considered is  $0 \leq \angle(C) < \pi/2$; $\angle(C)=\pi/4$ when $C$ is axis-parallel.

We define our hitting set using $p_0$ and five further points on $\Bcirc(p_0,1)$, but instead of $p_0=c(C)$, it is better now to have $p_0:=(t,0)$  for $t \in \mathbb{R}_+$  chosen later. 
We  fix $p_0:=(t,0)$, $p_1:=(t,1)$,  $p_3:=(t+1,0)$,  $p_5:=(t,-1)$, parameterized by $t$.  
 Figure~\ref{fig:chosing_t} shows the introduced hitting points with tentative choices for $p_2$ and $p_4$ for two possible values of $t$. The coordinates of $p_2$ and $p_4$ will need a more refined definition depending on the chosen $t$ and the angle $\angle(C)$.

The neighbours of $C$ having their centre in the region $\bigcup_{i=0}^5 B_2(p_i,\frac12)$ are hit by $\{p_0,\ldots, p_5\}$ (Proposition~\ref{prop:nu_alpha}). It remains to hit the squares not having their centre in these disks.
We call {\em holes} the ``triangular regions'' of $Q \setminus \bigcup_{i=0}^5 B_2(p_i,\frac12)$ (see Figure \ref{fig:chosing_t}).
In order to cover these holes as well, there are two difficulties to overcome:

\begin{figure}[ht]
	\centering
	\includegraphics[scale =0.65]{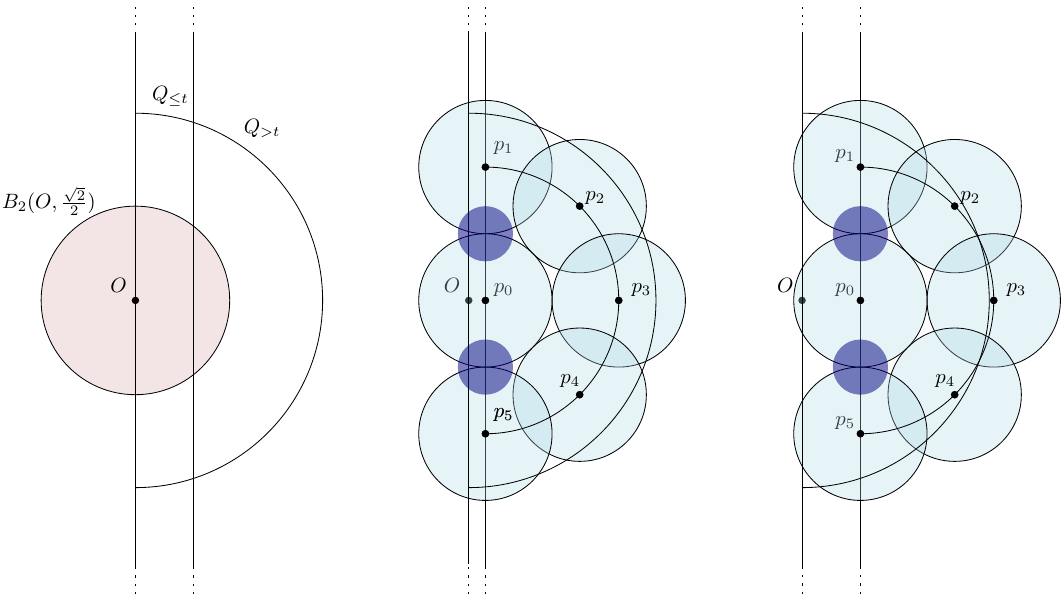}
	\caption{
      On the left, the disk $B_2(O, \frac{\sqrt{2}}2)$ swept by all possible $C$ and the half-disk $Q$ partitioned in $Q_{\leq t}$ and  $Q_{>t}$.
      In the centre and right, two possible hitting sets with $6$ points correspond to different parameter values $t$. For $i\in\{0,1,...,5\}$, each light (blue) disk $B_2(p_i, 1/2)$ contains all centres of possible unit squares whose inner disk is hit by its centre $p_i$ while the smaller dark (blue) disks represent some of the regions that are patched by Lemma~\ref{lem:third_ball}.}
    \label{fig:chosing_t}
\end{figure}

First, we have to cover the  holes of  $Q_{\leq t}:= \{(x,y) \in Q: 0\le x\le t\}$.
 The smaller $t$ is, the smaller the holes of  $Q_{\leq t}$ are. If $t$ is small enough, these holes can be ``patched'' by  Lemma~\ref{lem:third_ball}, as stated in Claim~2.
Second, for the squares with centre in  $Q_{>t} := \{(x,y) \in Q:  x > t\}$, if the $t$ value is too small, no matter how we fix $p_2,p_3,$ and $p_4$ on the half-circle $\{ (x,y) \in \overline{B}_2(p_0, 1)  :  x > t\}$ they will not suffice for satisfying condition $(iii)$ of Theorem~\ref{thm:in_out_convex_hull} for all possible $C$.    
We will see, though, that for the maximum value of $t$ computed in Claim~2,  there is a choice of these three points (depending on $\angle(C)$) such that $\{p_0, p_1, ..., p_5\}$ hits all the squares centred in $Q_{>t}$.

\smallskip
\noindent{\bf Claim 2}: If  $t=\frac{\sqrt{4\sqrt{2}-5}}{4}$, then $S\cap \{p_0, p_1, p_5\}\ne\emptyset$ for any unit square  $S$, $c(S)\in Q_{\leq t}$.
\begin{cpf}
Let $S\in N(C)$ be a square with centre $c(S)$ and assume for simplicity that $c(S)$ has a positive vertical coordinate, the other case being symmetric. We show that $\{p_0,p_1\}$ hits $S$ (in the symmetric case, $S$ is hit by $\{p_0,p_5\}$).
Since unit squares contain a disk of radius $\frac12$, this is true  if  $c(S)\in B_2(p_0,\frac12)\cup B_2(p_1,\frac12)$. 
Apply now Lemma~\ref{lem:third_ball} with $a=p_0$, $b=p_1$, and consequently $d=\frac12$, to conclude that any unit square with centre in $B_2(q, \frac{\sqrt{2}-1}2)$ contains either $p_0$ or $p_1$,  where $q$ is the midpoint of the segment $[p_0,p_1]$. Therefore, if the intersection point of $B_2(p_0,\frac12)$ with the vertical axis is contained in $B_2(q, \frac{\sqrt{2}-1}2)$, then $\{p_0,p_1\}$ intersects every square $S$ with $c(S)\in Q_{\leq t}$, provided that the following two conditions are satisfied:

\begin{equation*}
    t^2 + \left ( \frac12 -\sqrt{\frac14-t^2} \right )^2\le \left ( \frac{\sqrt{2}-1}2 \right )^2;\quad \dist (p_1, (0,\sqrt{2}))^2=t^2+(\sqrt{2}-1)^2\le \frac14.
\end{equation*}

The maximum of $t$ under the first condition is $\frac{\sqrt{4\sqrt{2}-5}}{4}$, and for this value, the second condition is also satisfied, finishing the proof of Claim~2.
\end{cpf}

Once the value of $t$ is fixed, we proceed by defining the remaining points of the hitting set. This definition will depend on the angle $\angle(C)$:

\smallskip
\noindent{\bf Claim~3}:  There exists $p_2\in\Bcirc(0,1)$ such that if $0 \leq \angle(C)\le \frac{\pi}4$, 
\begin{equation*}
    \dist(\frac{p_1+p_2}2,p_2) \le \dist(\frac{p_1+p_2}2,C),\quad \dist(\frac{p_3+p_2}2,p_2) \le \dist(\frac{p_3+p_2}2,C).
\end{equation*}
Similarly, there exist  $p_2'$, where the same holds if $\frac{\pi}4\le \angle(C) < \frac{\pi}2$,  replacing $p_2$ by $p_2'$.    

\begin{cpf}
To prove the claim, increase $\angle(C)$ from $0$ to $\frac{\pi}4$, continuously: the union of the points of the changing squares $C$ is denoted by $C_1$ (Figure~\ref{fig:C1_C2}$(a)$);  $C_1$ is a well-defined closed set.
Note also that moving the candidate for $p_2$ from $p_3$ to $p_1$ on $\Bcirc(p_0, 1)$, the segment $[p_2,p_3]$ increases, and   $[p_1,p_2]$ decreases. The two disks having these segments as diameters -- denote these open disks by $D_{23}, D_{12}$ --   also increase and decrease respectively. 

\begin{figure}[ht]
	\centering
	\includegraphics[scale =0.7]{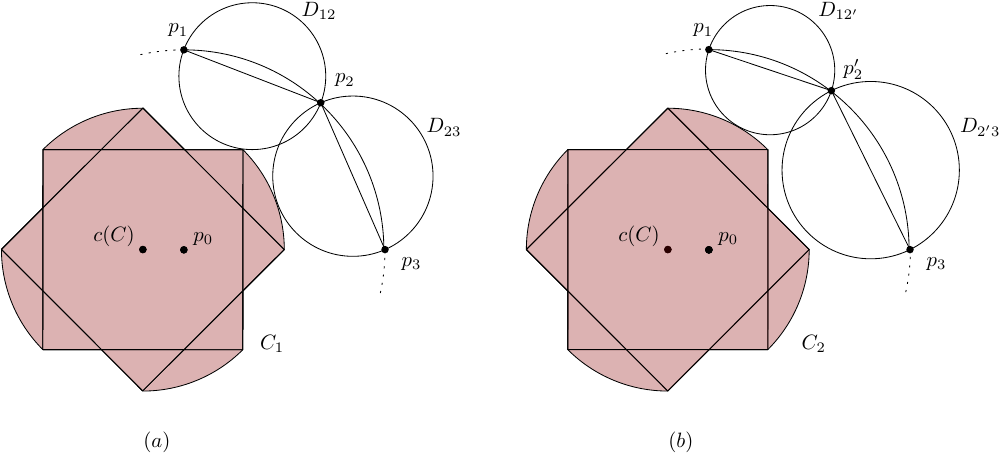}
	\caption{Representation of the sets $C_1$ and $C_2$.}
	\label{fig:C1_C2}
\end{figure}

The assertion of Claim~3 for squares of $\angle(C)\le \frac{\pi}4$ is then clearly equivalent to the following: 

{\em  There exists a point $p_2$ for which both $D_{23}$ and $ D_{12}$ are disjoint from $C_1$.} 

\noindent
Indeed, then 
\begin{align*}
&\dist(\frac{p_1+p_2}2,p_2) \le \dist(\frac{p_1+p_2}2,C_1)\le \dist(\frac{p_1+p_2}2,C), \\
&\dist(\frac{p_2+p_3}2,p_2) \le \dist(\frac{p_2+p_3}2,C_1) \le \dist(\frac{p_2+p_3}2,C).
\end{align*}
With the $t$ value provided by Claim~2, $p_2:= (t+\cos(0.82), \sin(0.82))$ verifies these properties. 

Similarly, increasing $\angle(C)$ from $\frac{\pi}4$ to  $\frac{\pi}2$ $C$ ``sweeps'' $C_2$, where \quotes{$\angle(C)=\frac{\pi}2$} actually means $\angle(C)=0$ with our formal definition (Figure~\ref{fig:C1_C2}$(b)$).
If   $p_2':= (t+\cos(0.92), \sin(0.92))$, the open disks with diameters $[p_2',p_3]$ and $[p_1,p_2']$ are both disjoint of $C_2$, finishing the proof of Claim~3.
\end{cpf}


Now, we can define the two sets hitting the neighbours of $C$ according to $\angle(C)$.
Recall  that we fixed  $p_1:=(t,1),  p_5:=(t,-1)$, $p_0:=(t,0)$, $p_3:=(t+1,0)$, where  $t=\frac{\sqrt{4\sqrt{2}-5}}{4}$, and we defined $p_2$ and $p_2'$ to satisfy  Claim~3 under two different conditions that cover all the  possibilities for $C$. Denote $p_4$ the reflection of $p_2$ to the horizontal axis, and $p_4'$ the reflection of $p_2'$.
We show that all neighbours of $C$ are hit by  $H_1 := \{p_0, p_1, p_2, p_3, p_4', p_5\}$ if  $0\le \angle(C)\le \frac{\pi}4$,  and by  $H_2 := \{ p_0, p_1, p_2', p_3, p_4, p_5\}$ if  $\frac{\pi}4\le \angle(C) <\frac{\pi}2$:

If $S\in N(C)$, then by Claim~1, $c(S)\in Q$. If furthermore $c(S)\in Q_{\leq t}$, then we get from Claim~2 that already the  $3$-element subset $\{p_0, p_1, p_5\}$ is hitting. 
Otherwise $c(S)\in Q_{> t}$. We partition $C$ into $C_{\leq t} := \{(x,y) \in C: x\le t\}$ and $C_{> t} := \{(x,y) \in C: x > t\}$. In Claim~4, we consider the case $S \cap C_{>t} \neq \emptyset$, and in Claim~5, the alternative case $S \cap C \subseteq C_{\leq t}$. This will conclude the proof of the lemma.

\smallskip
\noindent{\bf Claim~4}:  If a unit square $S \in N(C)$ satisfies $c(S) \in Q_{>t}$ and $S\cap C_t \not = \emptyset$, then $S$ is hit by  $H_1$ if $0 \leq \angle(C) \leq \frac{\pi}{4}$, and by $H_2$ if $\frac{\pi}{4} \leq \angle(C) < \frac{\pi}{2}$.
\smallskip 

\begin{cpf}
We assume $0 \leq \angle(C) \leq \frac{\pi}{4}$, since the case with $\frac{\pi}{4} \leq \angle(C) < \frac{\pi}{2}$ is the same by symmetry.
We define $P_1:=\conv(H_1)$, a hexagon with the peculiarity that two sides are collinear since $p_0$ is contained in $[p_1, p_5]$,
and apply Theorem~\ref{thm:in_out_convex_hull} considering $P = P_1$, $p_0 = p_0$, and $C = C_{>t}$.
Observe that $p_0$ plays the double role of corner of $P_1$ and \quotes{centre} in condition $(i)$.  
Clearly, $C_{>t} \subseteq P_1$ and conditions $(i)$ and $(ii)$ are satisfied. 
We continue by checking the other conditions in $(iii)$:

Note that $c(S)$ is on the same side of the line $L_5=L_0$ containing $[p_5,p_0]$, and $[p_0,p_1]$. Therefore, looking at the additional assertion of Theorem~\ref{thm:in_out_convex_hull} we need to check $(iii)$ only for the indices in  $I:=\{1, 2, 3, 4\}$.
Since $0 \leq \angle(C) \leq \frac{\pi}{4}$, Claim~3 makes sure that $(iii)$ holds for  $i=1$, and $i=2$. However,  the angle of the vertex of $C$ which is in the quadrant $\{(x,y) : x>0, y\leq 0\}$ is $|\angle(C)  - \frac{\pi}2|\geq \frac{\pi}{4}$, and by symmetry this luckily means that $(iii)$ holds for $i=3$ and $i=4$.  So it holds for all $i\in I$, and therefore the assertion of Theorem~\ref{thm:in_out_convex_hull} can be applied, that is, $S$ is hit by  $H_1 = \{p_0, p_1, p_2, p_3, p_4', p_5\}$ and the claim is proved.
\end{cpf}

\smallskip
\noindent{\bf Claim~5}:  If a unit square $S \in N(C)$ satisfies $c(S) \in Q_{>t}$ and $S\cap C \subseteq C_{\leq t}$, then $S$ is hit by  $H_1$ if $0 \leq \angle(C) \leq \frac{\pi}{4}$, and by $H_2$ if $\frac{\pi}{4} \leq \angle(C) < \frac{\pi}{2}$.

\begin{cpf}
Let $S \in N(C)$ be a square satisfying the hypothesis of the statement.
First, if $S$ intersect any of the segments $[a,b]$ for $a,b$ two consecutive points in $(p_1, p_2, p_3, p_4', p_5)$ if $0 \leq \angle(C) \leq \frac{\pi}{4}$, or in $(p_1, p_2', p_3, p_4,p_5)$ if $\frac{\pi}{4} \leq \angle(C) < \frac{\pi}{2}$, then we can apply Lemma \ref{lem:outside} and, since we verified already in Claim~4 that 
\begin{equation*}
    \dist(\frac{a+b}{2},a) \leq \dist(\frac{a+b}{2},C),
\end{equation*}
we may conclude that $S$ is hit by $H_1$ or $H_2$, depending on $\angle(C)$.
The same conclusion can be drawn if $c(S)$ is contained in $B_2(p_0, \frac{1}{2}) \cup B_2(p_5, \frac{1}{2})$.

Otherwise, assume that $S \cap \conv(\{p_1, p_2, p_3, p_4', p_5\}) = \emptyset$, or $S \cap \conv(\{p_1, p_2', p_3,$ $ p_4, p_5\}) = \emptyset$, depending on $\angle(C)$, and  $c(S) \not \in B_2(p_0, \frac{1}{2}) \cup B_2(p_5, \frac{1}{2})$.
Let $c$ be the point of $S\cap C \subseteq C_{\leq t}$ closest to $c(S)$.
In this case, we show that $\dist(c,c(S))> \frac{\sqrt{2}}2$. Since $c,c(S) \in S$, this contradicts that $S$ is a unit square.

Consider the segment $[c(S),c] \subseteq S$. It intersects the vertical line $l_t := \{(x,y) : x = t\}$ because $l_t$ separates the points $c$ and $c(S)$. Moreover, due to the fact that $S$ does not intersect the polygon $\conv(\{p_1, p_2, p_3, p_4', p_5\})$ if $0 \leq \angle(C) \leq \frac{\pi}{4}$, or $\conv(\{p_1, p_2', p_3, p_4,p_5\})$ if $\frac{\pi}{4} \leq \angle(C) < \frac{\pi}{2}$, this intersection is either above $p_1$ or below $p_5$. We can assume that it crosses $l_t$ above $p_1$ (the other case is symmetric).
Under this assumption, $c(S)$ has a vertical coordinate larger than $1$ and so the horizontal line $h := \{(x,y) : y=1\}$ separates $c$ and $c(S)$. We denote by $q$ the intersection point between $[c(S),c]$ and $h$ (Figure \ref{fig:cover_ct}).

\begin{figure}[ht]
    \centering
    \includegraphics[scale=0.8]{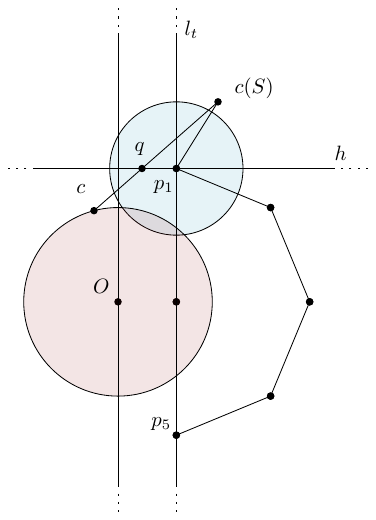}
    \caption{Representation of the centre of a square that is not hit by any point in $H_1$.}
    \label{fig:cover_ct}
\end{figure}

Now, it is easy to see that, $\dist(c(S),q) > \dist(c(S), p_1) > \frac{1}{2}$ where the second inequality comes from the assumption $c(S) \not \in B_2(p_1,\frac{1}{2})$. Moreover,  $\dist(q,c) \geq 1 - \frac{\sqrt{2}}{2}$ since $q\in h$ and $c \in C \subset B_2(O, \frac{\sqrt{2}}{2})$ lies below the horizontal line $\{(x,y) : y = \frac{\sqrt{2}}{2}\}$.
Finally, 
\begin{equation*}
    \dist(c(S),c) = \dist(c(S),q) + \dist(q,c) > \frac{1}{2} + 1 - \frac{\sqrt{2}}{2} > \frac{\sqrt{2}}{2}.
\end{equation*}
\end{cpf}

\subsubsection*{Proof of Theorem~\ref{thm:hitting}.}
Let $S$ be a square minimizing $l(S')$, $S'\in\Sscr$, and for the simplicity of the notations, choose $l(S')=1$ to be of unit length.

Now for each square $Q \in N(S)$, fix a point $p\in Q \cap S$ and let $Q'$ be a unit square containing $p$ and completely contained in $Q$ (see Figure \ref{fig:square_transformation}), and let $N':=\{Q': Q\in N(S)\}$. We say that $Q$ and $Q'$ {\em correspond} to one another. Clearly, each $Q'\in N'$ still intersects $S$, and for each point hitting a subset of $N'$, the same point hits all the corresponding sets in $N(S)$ (\emph{local homothety}, like in the proof of Theorems~\ref{thm:Pach} and~\ref{thm:hit_convexsets}).

Since $\{S\}\cup N'$ contains only unit squares, by Lemma \ref{lem:hitting_neighbours} $\tau(N[S]) \leq \tau(N') \leq 10$. Since we can iterate this procedure to any subfamily of $\Sscr$, we proved that $\Sscr$ is hitting-$10$-degenerate, and Lemma \ref{lem::induction_idea} concludes the proof of the first statement of Theorem~\ref{thm:hitting}. 
\begin{figure}[ht]
	\centering
	\includegraphics[scale =0.65]{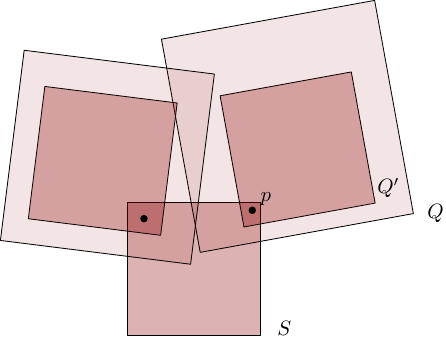}
	\caption{Transformation of the squares in $N(S)$.}
	\label{fig:square_transformation}
\end{figure}

The proof of the second statement is analogous and even simpler since now we directly have only unit squares, so by the 
second part of Lemma \ref{lem:hitting_neighbours}, there exists a square $S$ with the property $\tau(N[S]) \leq 6$. Therefore, $\Sscr$ is now hitting-$6$-degenerate, and we conclude by Lemma \ref{lem::induction_idea} again. 

\subsection{Lower bound for hitting}\label{sub:lower} 
We continue now with the proof of the lower bound stated in Theorem~\ref{thm:lower}:

\begin{proof}[Proof of Theorem~\ref{thm:lower}]
The proof describes a family of unit squares with $\nu=1$ and $\tau=3$ and one of squares of arbitrary sides with $\nu=1$ and $\tau=4$.
The constructions borrow ideas of \cite{2021_Biniaz} for unit disks or \cite{2011_Dumitrescu} for translations of a triangle. 

Start by Figure~\ref{fig:tauVSnu} (a), and note that the three squares of the figure {\em can be hit by two points only if one of these points is a vertex of the triangle formed by one side of each.}   

\begin{figure}[ht]
	\centering
	\includegraphics[scale =0.58]{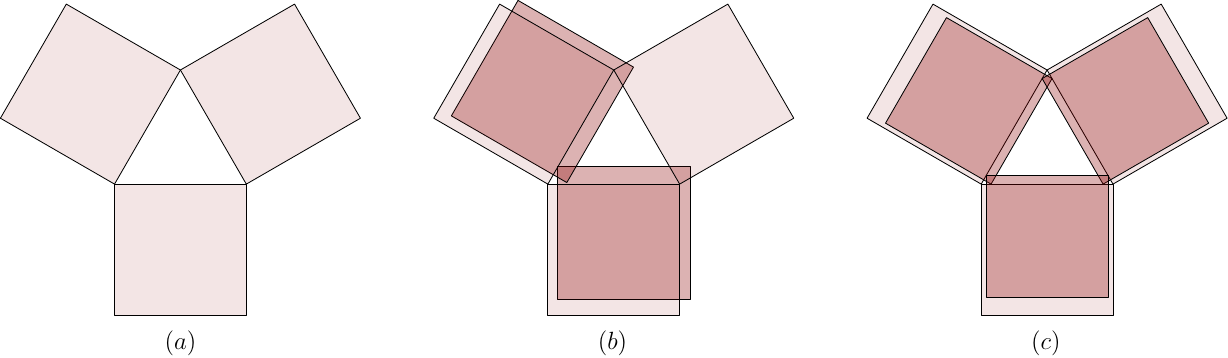}
	\caption{Pairwise intersecting squares, i.e., $\nu=1$, and $\tau=2$ (a), $\tau=3$  (b),   (c).}
	\label{fig:tauVSnu}
\end{figure}

Add now a small shift of each of the two squares at each vertex as  (b) shows for one of the three vertices. Six squares are added in this way, altogether, we mean a family of nine squares in the example (b). {\em We have  $\tau=3$ for this family, since deleting any vertex of the triangle (as proved to be necessary) we need two more vertices.} 

The same holds for the six squares of two different sizes with $\nu=1$ of (c), so we have $\tau=3$ for the same reason, and we use this to continue the construction with more squares for having $\tau=4$ while keeping $\nu=1$. 

\begin{figure}[ht]
	\centering
	\includegraphics[scale =0.58]{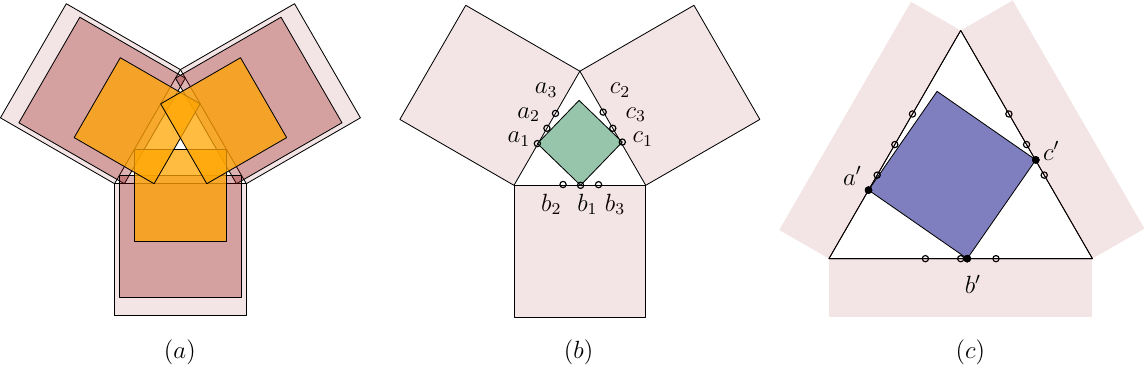}
	
	\caption{$13$ pairwise intersecting squares with $\nu=1$, and $\tau=4$.}
	\label{fig:tau4}
\end{figure}
To this end,  Figure~\ref{fig:tau4} adds three more layers to Figure~\ref{fig:tauVSnu} (c).
Figure~\ref{fig:tau4}~(a) contains   Figure~\ref{fig:tauVSnu}(a), but each of the three squares of the latter is completed to a chain of $3$ squares containing one another. From the largest to the smallest, we refer to these squares as the \emph{pink}, \emph{red}, and \emph{orange} square. Each colour forms a triangle, and the six squares belonging to any pair of colours form a drawing \quotes{isomorphic} to Figure~\ref{fig:tauVSnu}(c). (Mainly: the intersections of different pairs of squares of the same colour are disjoint from one another.) 
Therefore, the $6$ squares of any two colours cannot be hit by less than $3$ points. 
It follows that none of the intersection points of the largest pink squares can be in a hitting set of size three since it does not hit any smaller (red or orange) square. Therefore  {\em a hitting set of size three must contain one point of each pink square.}

Now add three smaller squares that will be referred to as the \emph{green} squares, touching each of the three pink squares at a point, one of them in the midpoint of a side. A green  square intersecting the pink squares in $\{a_1,b_1,c_1\}$ is drawn on Figure~\ref{fig:tau4}~(b).
The other two are symmetric and intersect the pink squares in $\{a_2,b_2,c_2\}$ and $\{a_3,b_3,c_3\}$, respectively.
Since we also have to hit these green squares, from the conclusion of the previous paragraph, we get that {\em the only hitting sets of size $3$ are $\{a_1, b_2, c_3\}$ and the \emph{six} different symmetric versions of it}.

Finally, we have obtained three squares of each of the four colours, and we proceed by adding a thirteenth square that intersects the pink squares at points $\{a',b',c'\}$, as shown in Figure~\ref{fig:tau4}~(c).
These three points can be chosen to be distinct from $a_1$, $b_1$, $c_1$, $a_2$, $b_2$, $c_2$,  $a_3$, $b_3$, and $c_3$.
Thus, the last added square intersects all the previous ones, but it is not hit by any of the six hitting sets described before, obliging a fourth point for the hitting set.
\end{proof}

Taking disjoint copies of the $13$ squares of Figure~\ref{fig:tau4}, and of the $9$ of Figure~\ref{fig:tauVSnu}~(b), we immediately obtain the following result:

\begin{corollary} There exists families of squares with arbitrarily large values of $\nu$ such that $\tau= 4\nu$, and also of squares of  equal size such that $\tau= 3\nu$.
\end{corollary}

\subsection{Colouring squares}\label{sub:colouring}

 In this section, we colour families of squares to bound their chromatic number for a proof of Theorem~\ref{thm:colouring}.
 The framework we use is a well-known averaging argument (see \cite{1960_Asplund, 2021_Chalermsook}), which calls for the simple statement of Lemma~\ref{lem:cross}. Our proof of this lemma turned out to require more effort than one could guess at first sight:
 
Let $R$ and $S$ be two squares in the plane, we say that $R$ and $S$ are \emph{crossing}
if $R\cap S$ is non-empty and does not contain any of the eight vertices of the two squares.

\begin{lemma}\label{lem:cross}
Given two crossing squares,  each of them contains the centre of the other.
\end{lemma}

\begin{proof} Let $R$ and $S$ be two crossing squares, $p\in R\cap S$, and suppose for a contradiction that $c(S)\notin R$. Then the segment $[p,c(S)]$ crosses the boundary of $R$ in a side of $R$, let $[u,v]$ be this side, where $u$ and $v$ are vertices of $R$, and $L$ the line that contains $[u,v]$. 
Since $S$ contains neither $u$ nor $v$,  $L\cap S$ is a non-empty  subsegment of the open interval $(u,v)$, denote its endpoints by $u'$ and $v'$.
We distinguish two  cases:

\begin{figure}[ht]
	\centering
	\includegraphics[scale =1]{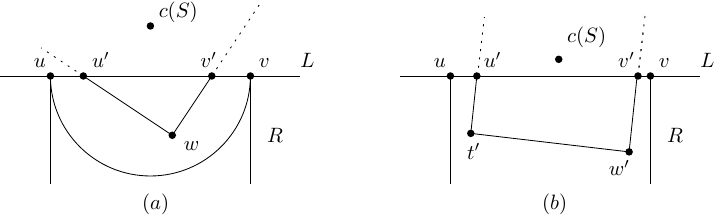}
	\caption{The two kinds of crossing intersections.
	}
	\label{fig:Lemma4}
\end{figure}

\smallskip
\noindent{\bf Case a}: 
{\em The segment $[u',v']$ joins two intersecting sides of $S$} (Figure \ref{fig:Lemma4}(a)).
\smallskip

Then the intersection of $S$ with the half-plane limited by $L$ and not containing $c(S)$ is a right-angled triangle $u'v'w$, where $w$ is a vertex of $S$.  So $w$ is separated from $c(S)$ by $L$, that is,  $w$  is on the same side of $L$ as  $R$. 

Now by the only if part of Proposition \ref{prop:Thales}  applied to $a=u$, $b=v$, and $c=w$, $w$ lies in the open half-disk with diameter $[u,v]$. However, this half-disk is fully contained in $R$, and contains the vertex $w$ of $S$,  contradicting that $R$ and $S$ are crossing.

\smallskip
\noindent{\bf Case b}: 
{\em The segment $[u',v']$ joins two parallel sides of $S$} (Figure \ref{fig:Lemma4}(b)).
\smallskip

Then $L$ separates $c(S)$ and two vertices $t'$ and $w'$ of $S$. Assume that $t'$ lies on the same side of $S$ as $u'$,  and $w'$ on the same side as $v'$. 
Both $t'$ and $w'$ are in the same half-plane bounded by $L$, as $R$, and also in the same half-plane as $R$, bounded by the line through the side of $R$ parallel to $L$, since otherwise  $\dist(u',t')\ge l(R)$ or $\dist(v',w')\ge l(R)$ respectively, contradicting $l(S)  \le \dist(u',v')< l(R).$ 

Furthermore, {\em at least one of two vertices $t'$ and $w'$ of $S$ must also be in the intersection of the half-spaces limited by the two parallel sides of $R$ perpendicular to $L$, containing $R$.} Then $t'$ or $w'$ is in $R$, contradicting that $R$ and $S$ are crossing. Indeed,  if $t'$ and $w'$ are in different half-spaces not containing $R$ then $\dist(t',w') = l(S)> l(R)$, which is the same contradiction again;  or if they are in the same half-space not containing $R$ then either $\dist(u',t')$ or $\dist(v',w')$ would be strictly larger than $\dist(u',v')\ge l(S)$.
\end{proof}

\subsubsection*{Proof of Theorem~\ref{thm:colouring}}
To prove the first part of Theorem \ref{thm:colouring}, we adapt the \emph{averaging argument} used by Asplund and Gr\"unbaum in \cite{1960_Asplund}, exhibited in an elegant way by Chalermsook and Walczak \cite[Lemma 3]{2021_Chalermsook}.

Let $\Sscr$ be a family of squares, $G:=G(\Sscr)$ the intersection graph of $\Sscr$, $n := |\Sscr|$. 
Each point of a square can be contained
in at most $\Delta(\Sscr) - 1$ other squares with some strict inequalities at the borders (for example, the left-most point that is a vertex of a square cannot be contained in any other square).

For each square $R\in \Sscr$ and $R' \in N(R)$, counting {\em twice}   the pairs $(v, R')$, if $v$ is a vertex of $R$ and  $v \in R'$   and {\em only once } if $v$ is the centre of $R$  and $v\in R'$,  we get at most $(2\times 4 + 1)$ times the maximum degree of these vertices minus one,  for each square. Adding these numbers for all squares,  the sum we get is strictly less than $9n(\Delta(\Sscr) - 1)$.

This sum counts each edge of $G$ at least twice because if two squares have a vertex-intersection, then there exists a pair $(v,R)$,  such that $v$ is a vertex of one of them, $R$ is the other and $v\in R$,  and this pair is counted twice; if they have a crossing intersection, then by Lemma \ref{lem:cross} applied twice, both centres are in the intersection. 
Hence, $2|E(G)| < 9n (\Delta(\Sscr) - 1)$,
and therefore the 
minimum degree of $G$ is strictly less than $9(\Delta(\Sscr) - 1)$. Applying this to any subgraph, we get that $G$ is $k$-degenerate with $k < 9(\Delta(\Sscr) - 1)$, and hence  $9(\Delta(\Sscr) - 1)$-colourable.    

To prove the second part of the theorem, let now $\Sscr$ be a family of unit squares.
Lemma~\ref{lem:hitting_neighbours} implies that $\Sscr$ is hitting-$6$-degenerate. Then, by Lemma ~\ref{lem:degimplication}, $\Sscr$ is also ($6\Delta(\Sscr)-1$)-degenerate and therefore $6\Delta(\Sscr)$-colourable.

%
\qed

\smallskip

Note that both parts of Theorem~\ref{thm:colouring} rely on degeneracy, applying two distinct methods though: an averaging argument for squares of different sizes and Lemma~\ref{lem:hitting_neighbours} specifically for unit squares.
While both techniques are applicable to both cases, the results differ in efficiency. For arbitrary squares, Lemma~\ref{lem:hitting_neighbours} yields an upper bound of $10\Delta$, whereas the averaging argument provides the tighter bound $9(\Delta - 1)$ of Theorem~\ref{thm:colouring}. However, when considering unit squares, the averaging argument does not straightforwardly improve the  $9(\Delta - 1)$ bound, whereas the more sophisticated Lemma~\ref{lem:hitting_neighbours} yields the better bound $6\Delta$ with the exception of $\Delta = 2$, when clearly $9(\Delta - 1) < 6\Delta$.

\section{Conclusion and open questions}\label{sec:end} 

In this paper, we provided the best linear bounds we could for the hitting number of squares in the plane: Theorem \ref{thm:hitting} provides the upper bounds, and Theorem~\ref{thm:lower} the lower bounds, both proved in Section~\ref{sec:proofs}. These establish the worst $\tau/\nu$ ratio in the interval $[3,6]$ for unit squares and in $[4,10]$ for squares in general. Finding the right value remains elusive even for the special case when $\nu = 1$, meaning that the squares are pairwise intersecting.
Gr\"{u}nbaum \cite{1959_Grunbaum} mentioned that any family of pairwise intersecting translates and rotations of a fixed convex set can be hit by a constant number of points, depending on the convex object.
For the particular case of unit disks, H. Hadwiger and H. Debrunner \cite{1955_Hadwiger} provided an exact result showing that any family of pairwise intersecting unit disks can be hit by 3 points and that 3 points are sometimes necessary.
For unit squares 3 points are necessary by 
Theorem~\ref{thm:lower}, and it is not difficult to show that $4$ points are sufficient.
We wonder if this bound can be improved or if a better lower bound can be achieved.
We ask the following questions:

\begin{question}
Can every family of pairwise intersecting unit squares be hit by $3$ points? 
\end{question}

\begin{question}\label{que:dancing_lower}
Is there a family of squares with $\frac{\tau}{\nu} > 4$?
\end{question}

We also wonder whether the best examples for Question 2 can be realized with $\Delta=2$.

The bounds for the hitting number we established are based on hitting-degeneracy  (Lemma \ref{lem::induction_idea}). The greedy local induction of this approach is almost tight:  Figure \ref{fig:limit_degeneracy} presents a unit square $S$ with $\tau(N[S]) = 7$; moreover, removing the two squares strictly to the left of $c(S)$ we have an example of a unit square with left-most centre and $\tau(N[S]) = 5$, which shows the limits of this method.
Nevertheless, there is no reason to think the bounds of Theorem \ref{thm:hitting} cannot be improved with other methods.

\begin{figure}[ht]
	\centering
		\includegraphics[scale =0.45]{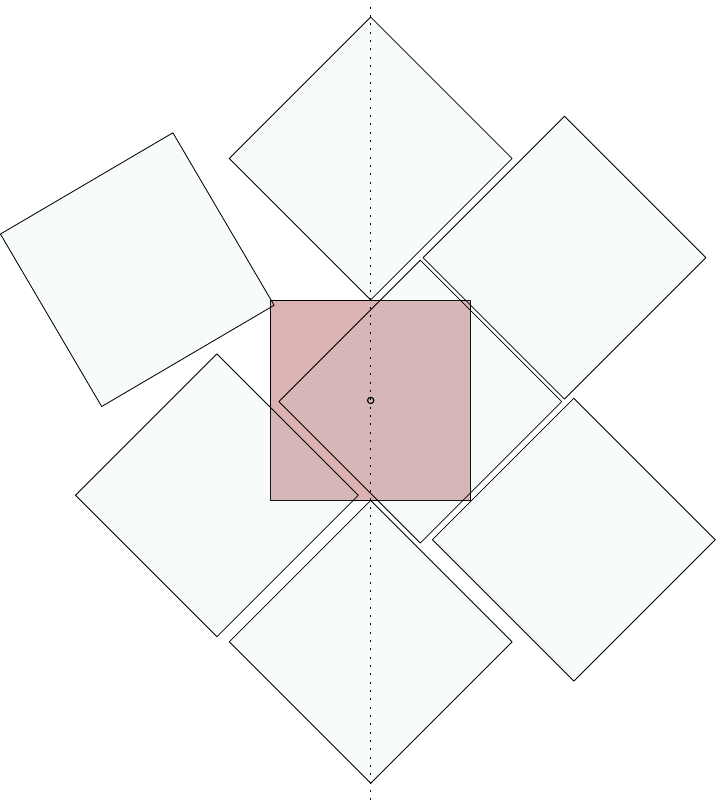}
		\caption{A unit square with seven pairwise disjoint neighbours.}
		\label{fig:limit_degeneracy}
\end{figure}

We would also like to draw attention to two open problems on axis-parallel squares. We pointed out in Table \ref{tab:tau_bounds} that $\tau \leq 4 \nu -3$ for axis-parallel squares and $\tau \leq 2 \nu - 1$ for axis-parallel unit squares. It is not known whether these two bounds are tight, and, surprisingly, the best-known lower bound is only $3/2$, simply achievable with a cycle of five unit squares.
We ask a question analogous to Question~\ref{que:dancing_lower} in the axis-parallel case:
\begin{question}\label{question:hit_squares}
Is there a family of axis-parallel squares with $\frac{\tau}{\nu} > \frac{3}{2}$?
\end{question}
In 1965, Wegner~\cite{1965_Wegner} conjectured that for any family of axis-parallel rectangles, $\tau \leq 2\nu - 1$.
Almost sixty years later, this conjecture is still open. Actually, no constant upper bound for the $\tau/\nu$ ratio for axis-parallel rectangles is known. Having only the weaker bound $\tau \leq 4\nu - 3$ for axis-parallel squares, Wegner's conjecture remains a frustrating challenge even in this special case.

If the Helly property does not hold, the chromatic number is often easier to upper bound by the maximum degree, not necessarily equal to the clique number. Nevertheless, the obvious inequality $\Delta \leq \omega$ immediately implies a bound in terms of $\omega$ as well.
We wonder how big the gap between $\omega$ and $\Delta$ can be for squares:

\begin{question}
What is the maximum size of a family of pairwise intersecting squares (or unit squares) with $\Delta=2$?
\end{question}

For squares, analogous questions to Question~\ref{question:hit_squares} can be asked about the chromatic number as well. 
In the literature, we could not find any lower bound on the ratio $\frac{\chi}{\nu}$ better than $\frac{3}{2}$ neither for translations and homotheties of a fixed convex set nor for their translates and rotations. Families of axis-parallel squares and (not necessarily axis-parallel) unit squares are two particular examples of these families.
We ask the following questions:

\begin{question}
Is there a family of axis-parallel squares, or not necessarily axis-parallel unit squares, with $\frac{\chi}{\omega} > \frac{3}{2}$?
\end{question}

Finding the right value of $\sup \frac{\chi}{\omega}$ for squares is also an open problem in the particular case $\omega=2$, that is, when the intersection graph is triangle-free. Perepelitsa \cite[Corollary 8]{2003_Perepelitsa} showed that any triangle-free family of axis-parallel squares is $3$-colourable. Her result follows directly from Gr\"otzsch's theorem \cite{1959_Grotzsch}, once observed that the intersection graph of such a family is planar. Allowing the squares to rotate, this property is lost, but this change may not significantly impact the chromatic number.
A method similar to the one used by Perepelitsa \cite{2003_Perepelitsa}  can be used to prove 6-colourability. 

\begin{question}\label{question:col_triangle_free}
What is the smallest $k$ such that any triangle-free family of squares is $k$-colourable?
\end{question}

``Triangle-free'' means $\omega\le 2$ here. According to Theorem~\ref{thm:colouring}, the answer to Question~\ref{question:col_triangle_free} is between $3$ and $9$ under the weaker condition  $\Delta\le 2$.

\section*{Acknowledgements}

Marco Caoduro was supported by a Natural Sciences and Engineering Research Council of Canada Discovery Grant [RGPIN-2021-02475].

\bibliographystyle{plainurl}
\bibliography{references}


\appendix

\section{NP-hardness for axis-parallel unit squares}\label{sub:unitsquares} 

The main problems we study in this paper are already NP-hard for the simplest objects, axis-parallel unit squares, justifying approximation algorithms and bounds between packing and hitting numbers or clique and the chromatic numbers instead of computing these numbers exactly or proving min-max theorems. We started the introduction with pointers to most of the hardness results, but we did not find any reference for the NP-hardness of the colouring problems or the maximum independent set problem of unit squares. This Appendix fills this gap. 

We summarize the missing statements in the following theorem.  The proofs are easy using the ideas for unit disk graphs of Clark, Colburn, and Johnson \cite{1990_Clark} further developed by Gr\"af, Stumpf and Wei{\ss}enfels \cite{1998_Graf},   because the graphs used by these articles can be represented as intersection graphs of axis-parallel unit squares as well.

\begin{theorem}
	The $k$-colourability  of axis-parallel unit squares is NP-complete for any fixed $k\in\mathbb{N}, k\ge 3$, and so is the existence of an independent set of unit squares larger than a given size.  
\end{theorem}

\smallskip\noindent {\em Proof.}  
Let $k\in \mathbb{N}, k\ge 3$ and let us reduce first the $k$-colourability of a graph -- a well-known NP-complete problem \cite{1976_Garey} -- to the $k$-colourability of axis-parallel unit squares.
The main gadget of the reduction is a particular family of squares used to replace the edges of the graph. We introduce it and then overview its application.

Note that the ``chain of diamonds" of the graph in Figure~\ref{fig:forcing_chain}~(a) forces the endpoints $x$ and $y$ to have the same colour in any $3$-colouration. This graph can be straightforwardly realized with unit squares forming a horizontal or vertical ``strip''  as long as necessary (Figure~\ref{fig:forcing_chain}~(b)). Furthermore, in Figure~\ref{fig:forcing_chain}~(a) each thick (red) edge in the middle can be replaced by a $(k-1)$-clique which is the intersection graph of Figure~\ref{fig:forcing_chain}~(b) with the dark (red) squares in the middle replaced by $k-1$ squares with the same intersections and intersecting one another (it can be $k-1$ copies of the same square like on the right of Figure~\ref{fig:forcing_edge}). We will call this set of squares a  {\em $k$-chain of unit squares}. 

Furthermore, for an edge of a graph drawn in the plane alternating between horizontal and vertical segments, not crossing any other edge,  each horizontal and vertical segment can be replaced by a $k$-chain, and the endpoint of a chain can be identified by the starting point of the next one as in  Figure~\ref{fig:forcing_edge}.    The obtained graph is still an {\em $k$-chain}.
Clearly,  \emph{the endpoints $x$ and $y$ of a $k$-chain   have the same colour in any $k$-colouration.}

\begin{figure}[ht]
	\centering
	\includegraphics[scale=0.8]{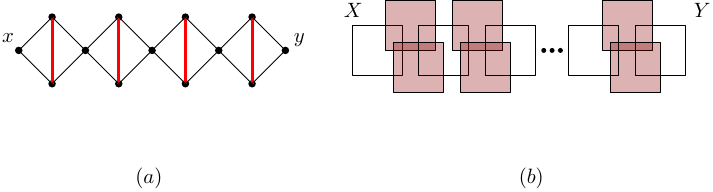}
	\caption{(a) A $3$-chain forcing vertices $x$ and $y$ to have the same colour in every $3$-coloration  (b) Unit square representation of a $3$-chain.}
	\label{fig:forcing_chain}
\end{figure}

Add now a point $z$ to this figure, and join it only to $y$: we know then that $x$ and $z$ must have different colours in any colouration (see Figure~\ref{fig:forcing_edge}). A family of unit squares whose intersection graph is a $k$-chain with such an additional vertex and edge will be called an {\em $k$-chain-edge of unit squares}. 

\begin{figure}[ht]
	\centering
	\includegraphics[scale=0.8]{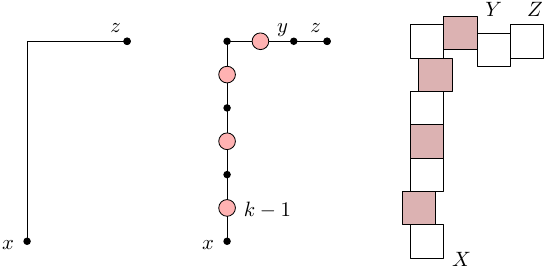}
	\caption{From left to right: an edge of $G$, a $k$-chain-edge, and a $k$-chain-edge of unit squares.}
	\label{fig:forcing_edge}
\end{figure}

Having noticed the realizability of $k$-chains as intersection graphs of unit squares, we can now mimic the proof of \cite{1990_Clark} or, more generally, of \cite{1998_Graf} to finish the proof: 

In the same way as \cite{1990_Clark} does for unit disks, the $3$-colourability of a planar graph with maximum degree $3$ -- already an NP-hard problem \cite{1976_Garey} -- can be reduced to the $3$-colourability of unit square graphs in polynomial time, using $3$-chains of unit squares. Indeed, according to \cite{1981_Valiant}, each planar graph $G$ of maximum degree three can be drawn in the plane without crossing edges in such a way that each edge consists of at most four horizontal or vertical sequences. This embedding can be realized in polynomial time \cite{1998_Bieldi}. It follows then easily from the above observations that each edge can be replaced by a $3$-chain-edge of unit squares, and the $3$-colourability of the constructed family of unit squares is equivalent to that of $G$.

Now, for arbitrary  $k\ge 3$ (also allowing equality, i.e., the $k=3$ case is reconsidered, the above discussion serves merely the pedagogical reasons for introducing some of the ideas and the main square-specific gadgets) and graph $G$, we reduce the colouring problem of $G$ to the colouring problem of unit squares. 

For this, the essential part of the work has already been done by Gr\"af, Stumpf and Wei{\ss}enfels \cite{1998_Graf}: they show that a graph $G'=(V'E')$ drawn in the plane can be constructed from $G$ in polynomial time (we will use it as a black box, but for an approximate illustration see Figure~\ref{fig:reduction}), where $V'$ is partitioned into a path $P_v$ for each $v\in V(G)$, a set $T$ of degree two vertices, and sets  $V_u$ indexed by the elements $u\in U$ of a set $U$ verifying the following properties:  
	
\begin{itemize}
	\item [(i)] Contracting  $G_u$ (Figure~\ref{fig:reduction} left)  for each $u\in U$,  the new vertices, denote them $x_u$ $(u\in U)$  are of degree $4$, and all the edges of the obtained graph are horizontal or vertical. Denote by $H$ the resulting graph drawn in the plane (Figure~\ref{fig:reduction} right). 
	
	 	\item [(ii)]If in $H$ all edges of the path $P_v$ $(v\in V(G))$ are contracted for each $v\in V(G)$,  furthermore, {\em splitting off} the two edges incident to $t\in T$, that is, replacing them by one edge between the two endpoints different from $v$, and doing the same with the opposite edges incident with $x_u$ $(u\in U)$ in the planar drawing of $H$ (and therewith losing planarity), we get back $G$. 
	 	
	 	Each edge $e=xz\in E(G)$ $(x,y\in V(G))$ corresponds to a path between $x$ and $z$ in $H$ whose edges are a set  $F_e$ of edges in $G'$. 

\item [(iii)] Replacing each horizontal and each vertical line of $G'$ by a $k$-chain, but for one edge in each $F_e$ a $k$-chain-edge, the constructed graph $G''$  is $k$-colourable if and only if $G$ is $k$-colourable.
\end{itemize}

\begin{figure}[ht]
	\centering
	\includegraphics[scale=0.9]{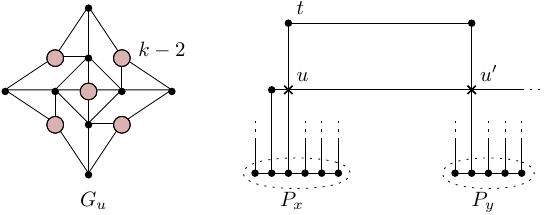}
	\caption{On the left, the graph isomorphic to all the graphs $G_u$ $(u\in U)$, the (red) disks represent $k-2$ cliques. On the right, a representation of a portion of $H$. The vertices $x,y \in V(G)$ are replaced with two paths $P_x, P_y \subset V'$ of size $d_G(x)$ and $d_G(y)$, respectively.}
	\label{fig:reduction}
\end{figure}


Figure \ref{fig:reduction} illustrates the graphs that make possible to prove these properties. 

Then \cite{1998_Graf} continue by realizing $k$-chains and $k$-chain-edges as intersection graphs of unit disks, and substitute these for the edges of $H$. By $(i)$ all edges of $H$ are horizontal or vertical, and this ensures that the substitution can be done. 

Now, $k$-chains and $k$-chain-edges of unit squares replace unit disks (Figure~\ref{fig:forcing_chain}); the endpoint of a chain can be identified by the starting point of the next one, again as in  Figure~\ref{fig:forcing_edge}, and clearly, this does not take more computational time than the disk version. Thus, a unit square graph is constructed in polynomial time, whose $k$-colourability   is equivalent to that of $G$.\qed

\end{document}